\documentclass[acmsmall,screen,nonacm]{acmart}

\usepackage{etoolbox}

\newtoggle{ext} 
\settoggle{ext}{true} 
\newtoggle{anon} 
\settoggle{anon}{false}

\makeatletter
\iftoggle{ext}{
  \setcopyright{none}
  \AtBeginDocument{
    \settopmatter{printccs=true,printacmref=true}
  }
  \renewcommand{\@mkbibcitation}{%
  \bgroup \par\medskip\small\noindent{\bfseries This article is an extended version of:}\par\nobreak
      \noindent
      Ruotong Cheng and Azadeh Farzan. 2025.
      Products of Recursive Programs for Hypersafety Verification.
      \textit{Proc. ACM Meas. Anal. Comput. Syst.} 9, OOPSLA2, Article 382 (October 2025), 29~pages.
      \url{https://doi.org/10.1145/3763160}
    \par\egroup }
}{
  \usepackage{xr}
  \externaldocument{app}
  \setcopyright{cc}
  \setcctype{by-nc-sa}
  \acmDOI{10.1145/3763160}
  \acmYear{2025}
  \acmJournal{PACMPL}
  \acmVolume{9}
  \acmNumber{OOPSLA2}
  \acmArticle{382}
  \acmMonth{10}
  \received{2025-03-25}
  \received[accepted]{2025-08-12}
}
\makeatother

\usepackage{extarrows,proof}
\usepackage{hyperref}
\usepackage{stmaryrd}
\usepackage{algorithm}
\usepackage[noend]{algpseudocode}
\usepackage{bm}
\usepackage{wrapfig}
\usepackage{picins,xcolor}
\usepackage{listings}
\usepackage[final]{showlabels}
\usepackage{subcaption}
\usepackage{mathpartir}
\usepackage{ebproof}

\renewcommand{\vec}[1]{\bm{#1}}

\DeclareMathOperator{\red}{red}
\DeclareMathOperator{\Lin}{TO}
\DeclareMathOperator{\PO}{PO}

\DeclareMathOperator{\stack}{Stack}
\DeclareMathOperator{\param}{param}
\DeclareMathOperator{\fout}{output}
\DeclareMathOperator{\stktop}{top}

\DeclareMathOperator{\letters}{alph}
\newcommand{\enabled}{\mathsf{enabled}}

\newcommand{\Lang}{\mathcal{L}}
\newcommand{\shuffle}{\parallel}

\newcommand{\call}{\ensuremath{\mathtt{call}}}
\newcommand{\ret}{\ensuremath{\mathtt{ret}}}

\DeclareMathOperator{\dec}{dec}

\usepackage{mathtools}
\DeclarePairedDelimiter{\llrr}{\llbracket}{\rrbracket}
\newcommand{\sem}[2]{\mathcal{#1}\llrr{#2}}

\newcommand{\tS}{\widetilde{\Sigma}}
\newcommand{\powerset}{\mathcal{P}}

\newcommand{\Sl}{\Sigma^{\mathsf{int}}}
\newcommand{\Sc}{\Sigma^{\mathsf{call}}}
\newcommand{\Sr}{\Sigma^{\mathsf{ret}}}
\newcommand{\dl}{\delta^{\mathsf{int}}}
\newcommand{\dc}{\delta^{\mathsf{call}}}
\newcommand{\dr}{\delta^{\mathsf{ret}}}
\usepackage{fancybox}
\newcommand{\letter}[1]{\ovalbox{#1}}

\def\marrow{{\marginpar[\hfill  $\rightarrow$]{\color{red}$\leftarrow$}}}
\def\azadeh#1{{\color{blue}({\sc Azadeh says: }{\marrow\sf #1})}}
\def\ruotong#1{{\color{teal}({\sc Ruotong says: }{\marrow\sf #1})}}
\def\azadeh#1{}
\def\ruotong#1{}

\DeclareFontEncoding{LS1}{}{}
\DeclareFontSubstitution{LS1}{stix2}{m}{n}
\newcommand{\wn}{\mathsf{wn}}
\newcommand{\wnshuffle}{\mathbin{\pzigzag}}

\DeclareRobustCommand{\pzigzag}{%
  \mathbin{\text{\usefont{LS1}{stix2frak}{m}{n}\symbol{"CA}}}%
}

\def\tool{{\sc HyRec}}

\def\nP{\widehat{P}}
\def\pre{\text{\sf pre}}
\def\post{\text{\sf post}}

\newcommand{\nestedconcat}{\mathsf{nested\_\!concat}}
\newcommand{\ord}{\mathsf{ord}}
\newcommand{\blue}[1]{{\color{blue}#1}}

\newcommand{\indep}{\mathsf{indep}}
\newcommand{\dep}{\mathsf{dep}}
\newcommand{\stcomp}[1]{{#1}^\mathsf{c}}

\newcommand{\id}[1]{%
  \ensuremath{\mathit{\mathcode`\-=`\-\relax#1}}}
\newcommand{\attrib}[2]{\ensuremath{#1.\hspace*{1pt}\id{#2}}}

\definecolor{myblue}{rgb}{0.120137,0.407711,0.709558}
\definecolor{mypink}{rgb}{0.779247,0.000000,0.389398}

\iftoggle{ext}{%
\newcommand{\refapp}[1]{Appendix \ref{#1}}
}{%
\iftoggle{anon}{%
\newcommand{\refapp}[1]{\cite[Appendix \ref*{#1}]{anon-ext}}
}{%
\newcommand{\refapp}[1]{\cite[Appendix \ref*{#1}]{ChengF2025}}
}
}

\begin{document}

\title{Products of Recursive Programs for Hypersafety Verification}
\iftoggle{ext}{
  \subtitle{(Extended Version)}
}{}
\author{Ruotong Cheng}
\orcid{0009-0004-1857-7251}
\affiliation{%
  \institution{University of Toronto}
  \city{Toronto}
  \country{Canada}
}
\email{chengrt@cs.toronto.edu}

\author{Azadeh Farzan}
\orcid{0000-0001-9005-2653}
\affiliation{%
  \institution{University of Toronto}
  \city{Toronto}
  \country{Canada}
}
\email{azadeh@cs.toronto.edu}


\begin{abstract}
  We study the problem of {\em automated hypersafety verification} of {\em infinite-state recursive programs}. We propose an {\em infinite} class of {\em product programs},  specifically designed with {\em recursion} in mind, that reduce the hypersafety verification of a recursive program to standard safety verification. 
  For this, we combine insights from {\em language theory} and {\em concurrency theory} to propose an algorithmic solution for constructing an infinite class of recursive product programs. One key insight is that, using the simple theory of {\em visibly pushdown languages}, one can maintain the {\em recursive structure} of syntactic program alignments which is vital to constructing a new product program that can be viewed as a classic recursive program --- that is, one that can be executed on a  single stack. Another key insight is that techniques from {\em concurrency} theory can be generalized to help define product programs based on the view that the parallel composition of individual recursive programs includes all possible {\em alignments} from which a {\em sound} set of alignments that faithfully preserve the satisfaction of the hypersafety property can be selected. On the practical side, we formulate a family of parametric canonical product constructions that are intuitive to programmers and can be used as building blocks to specify recursive product programs for the purpose of relational and hypersafety verification, with the idea that the right product program can be verified automatically using existing techniques. We demonstrate the effectiveness of these techniques through an implementation and highly promising experimental results. 
\end{abstract}



\maketitle

\section{Introduction}\label{sec:intro}
Hyperproperties \cite{ClarksonS2010} are properties expressible over multiple runs of a program. A $k$-safety property is a property whose violation is witnessed by $k$ finite-length program runs. Determinism is an example of such a property: non-determinism can only be witnessed by two runs of the program on the same input that produce two different outputs. This makes determinism an instance of a $2$-safety property. A hypersafety property is a $k$-safety property for $k > 1$.

Verification of hypersafety properties was first achieved by reducing the problem to a standard safety verification by means of {\em self-composition} \cite{BartheDR2004}. For example, in order to verify that $P$ is deterministic, one can verify that $P;P'$ under the assumption of equivalent inputs for $P$ and $P'$ guarantee the equivalence of their outputs at the end, where $P'$ is a fresh copy of $P$ with a disjoint memory. It was argued in \cite{BartheDR2004} that any other {\em independent composition} operator can be used in place of the sequential composition operator.  Later, {\em cross products} \cite{ZaksP2008}, which combine $P$ and $P'$ synchronously in lockstep, were used in a more limited setup of compiler validation. {\em Program products} \cite{BartheCK2011} were introduced as combining  the two notions where one can {\em fuse} the two programs both synchronously and asynchronously where appropriate. 

It is now  well-established that product programs are vital to verification of hypersafety properties and relational reasoning in general \cite{SousaD2016,ChurchillPSA2019,FarzanV2019,ShemerGSV2019,ItzhakySV2024}. In algorithmic verification, they permit the use of existing tools and techniques for classic safety verification of sequential programs to be applied to hypersafety verification. This is mainly due to the fact that they accommodate the creation of instances that admit simpler proofs. Approaches to hypersafety or relational verification typically {\em implicitly} or {\em explicitly} define a state space of product programs that may be used to simplify  verification. For instance in \cite{BartheCK2011,SousaD2016}, a successful derivation of the underlying logical rules {\em implicitly} defines the verified product program, while in \cite{FarzanKP2022}, the product programs are syntactically defined and look like classic programs. 
With the implicit construction, since the choice of application of a proof rule and the next alignment step are combined and one typically has control over the entire verification stack, one can opt for an ergonomic design that works more effectively; e.g., heuristics for proofs search can be tuned to the specific problem. On the other hand, since the explicit construction decouples the construction of the program from the construction of the proof, it can benefit from improvements in standard verification tools. More importantly, having access to a concrete (product) program opens up possibilities for using other standard tools for program analyses such as symbolic/concolic execution, fuzzing, monitoring and/or repair, and static analyses. The two sets of techniques are therefore complementary and serve different use cases. In this paper, we are interested in the {\em explicit} construction of product programs.

This problem is somewhat well-studied for iterative programs, but has only been lightly studied in the context of recursive programs, and the majority of the proposed techniques fit in the {\em implicit} construction category. Most even bypass the program view entirely and considers the question of how a similar concept can be incorporated in the CHC (Constrained Horn Clauses) encoding of them, in the form of {\sc CHCproduct} \cite{MordvinovF2017} or {\em fold/unfold} \cite{DeAngelisFPP2016,AngelisFPP2018a} techniques to implicitly explore the space of the product programs and proofs simultaneously (more details in Section \ref{sec:related-work}).  The very few \cite{EilersMH2018} that construct explicit product programs use specialized constructions that help solve a limited set of verification instances.

\begin{wrapfigure}[7]{r}{0.38\textwidth}\vspace{-10pt}
\includegraphics[scale=0.62]{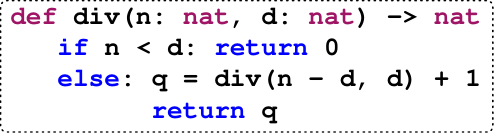}\vspace{-10pt}
\caption{Recursive integer division. \label{fig:iexample}}
\end{wrapfigure}
Consider the recursive code illustrated in Figure \ref{fig:iexample}, which computes
the integer result of dividing the numerator $\mathtt{n}$ by the denominator $\mathtt{d}$. Assume that we want to check if this {\em implementation} satisfies the following property of the original mathematical function:
\begin{align}
n \le n' &\implies \mathtt{div}(n,d) \le \mathtt{div}(n',d) \label{p1} \tag{\sc monotonicity}
\end{align}
First, observe that \ref{p1} is a 2-safety property.
The most straightforward encoding for solving the \ref{p1} property uses self-composition, that is, it reduces the problem to the safety verification of $\mathtt{div(n,d)};\mathtt{div(n',d')}$. However, this instance cannot be verified by an existing automated solver due to the nonlinearity of the arithmetic reasoning involved in guessing the required inductive function summaries \cite{GodefroidNRT2010}. In particular, the summary {\em must} precisely capture that $\mathtt{div(n,d)}$ is the {\em integer division} function.

Product programs are conceptually straightforward when it comes to non-recursive code. However, recursion introduces significant additional complexity. If $P$ and $P'$ are recursive, then $P;P'$ is a recursive program in the standard sense. However, the same cannot be said for an arbitrary product program made from $P$ and $P'$ (even if $P'$ is a copy of $P$). For instance, it is well-known that the parallel product of two recursive programs is a 2-stack machine, whose syntactic traces form a {\em context sensitive} language \cite{HopcroftMU2007}, which does not correspond to a recursive program. This problem does not exist in the iterative setting.

For \ref{p1}, the programmer's intuition points to a {\em lockstep} (or synchronous) product program, namely one in which the two recursive codes are executed {\em in tandem}, when possible. An appropriate product program then is one that (1) looks like a classic (single stack) recursive program, and (2) simulates the two copies of $\mathtt{div}$ in tandem. This product program can be defined using four recursive functions: $\mathtt{div1}$ and $\mathtt{div2}$ as two distinct copies of $\mathtt{div}$ from Figure \ref{fig:iexample}, and two distinct copies of the variation illustrated in Figure \ref{fig:s-prod}. In a lockstep simulation we need to execute in tandem while both copies last, and then a given copy alone when the other copy has terminated.  $\mathtt{div12}$ assumes that the second copy is still running and as such calls it in tandem, while $\mathtt{div1}$ assumes that the second copy has terminated and only locally recurses.

\begin{figure}[h]
\begin{center}
\includegraphics[scale=0.47]{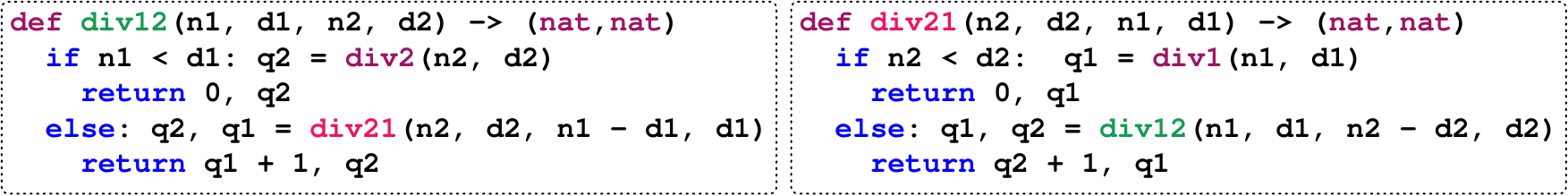}\vspace{-10pt}
\caption{With two disjoints copies of $\mathtt{div}$, these functions simulate a synchronous product of two $\mathtt{div}$'s}
\label{fig:s-prod}
\end{center}
\end{figure}

This product program does not syntactically match the one that we formally define and algorithmically construct in this paper, but it is more human-readable and conveys the two main ideas: (1) our product programs are {\em syntactic} in the sense that any syntactic run of it is identical to a syntactic run of the full (parallel) product of two distinct copies of $\mathtt{div}$ (not strictly true for this simplified example), and (2) even in the simplest case, they are far from trivial.  How can one systematically construct such a program? 

In \cite{FarzanV2019,FarzanV2020,FarzanKP2022}, this is achieved for iterative (non-recursive) programs, by taking an {\em operational} view of the program \cite{Farzan2023}. The set of all possible alignments are defined by the parallel composition of the components, which include arbitrary interleavings of syntactic runs from individual components. Then, a particular choice of product programs becomes about selecting a {\em sound} subset of these interleavings. As long as the subset has a finite representation (in a program or program-like structure), the product program and this set of behaviours are interchangeable.

Our first important observation is that, for the same operational style to work for recursive 
\begin{wrapfigure}[10]{r}{0.1\textwidth}\vspace{-12pt}
\includegraphics[scale=0.5]{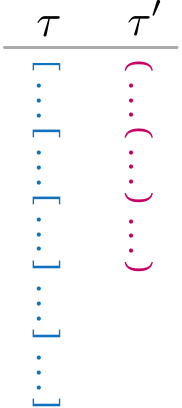}\vspace{-10pt}
\end{wrapfigure}
programs, one has to be explicitly conscious of the {\em structure} of the recursive runs, in terms of the 
matching of calls and returns in each component. To visualize this, consider runs $\tau$, $\tau'$ illustrated on the right, as being runs of the two disjoint copies $\mathtt{div(n,d)}$ and $\mathtt{div(n',d')}$ respectively. Open and closed brackets respectively denote calls and returns, and all internal actions are represented as the ellipses, since their specifics are unimportant. Note that each illustrated run is {\em well-nested}, in the sense that we have a {\em balanced} sequence of calls and returns. These are also called \emph{interprocedurally realizable paths} \cite{RepsHS1995,Ramalingam2000}.

Below we can see two (of the many) ways that one can  {\em interleave} (or align) $\tau$ and $\tau'$:
\begin{align} 
& \mathtt{\color{myblue} [ \dots\  [ \dots\   } \mathtt{\color{mypink} ( \dots\  ( \dots\  } \mathtt{\color{myblue} [ \dots}  \mathtt{\color{myblue} \ \dots ]  }
\mathtt{\color{mypink}\ \dots ) \ \dots )  } \mathtt{\color{myblue} \ \dots ]\ \dots ]  }
 \tag{\sc well-nested} \\
& \mathtt{\color{myblue} [ \dots\ [ \dots\  } \mathtt{\color{mypink} ( \dots\ ( \dots\  } \mathtt{\color{myblue} [ \dots   }
 \mathtt{\color{mypink}\ \dots)\ \dots )  } \mathtt{\color{myblue} \ \dots ]\ \dots ]\ \dots ] }
 \tag{\sc ill-nested}
\end{align}
and we contend that the well-nested run belongs to a product program and the ill-nested one does not. In the well-nested one, the matched calls and returns belong to the same component from which they originated. 
If we ignore the types/colours, we can pretend that the former is a well-matched run of a classic recursive program, whereas in the latter, we obtain an {\em interprocedurally unrealizable} path: a call and a non-matching return. Therefore, to construct a product program in this operational style, one needs (1) representations for syntactic program runs that maintain the recursive structure, and (2) mechanisms to isolate sound appropriate subsets of the well-nested alignments as product programs.

For (1), we lean on the well-established theory of {\em visibly pushdown languages} \cite{AlurM2004} to formally represent the syntactic runs of recursive programs, while retaining the recursive structure in the runs. Unlike a {\em call graph}, a classic model for recursive program analysis, a visibly pushdown {\em grammar} or {\em automaton} captures the {\em well-matched} syntactic runs of a recursive program. We then give a construction of the product program in the same formalism, which guarantees that it can be treated as a classic recursive program for all intents and purposes. 


To understand the complications of devising (2), let us further expand our simple $\mathtt{div}$ example with a slightly more complicated  property:
\begin{align}
n = 2 \times n' &\implies \mathtt{div}(n,d) \ge 2 \times \mathtt{div}(n',d) \label{p2} \tag{\sc scaling}
\end{align}
where \ref{p2} is a 2-safety property. Observe that a {\em lockstep} product is no longer appropriate. For \ref{p2}, we expect the execution of the first copy to be twice as long as the second copy. A {\em suitable} product program, therefore, must synchronize each recursive call of $\mathtt{div(n,d)}$ against two recursive calls of $\mathtt{div(n',d')}$. Consider a variation of this property where the scale factor  changes from $2$ to a different constant, like $5$, which immediately changes the product program needed for its verification to a new one. This motivates the following research question: how can one systematically define a state space of useful product programs for recursive programs for the purpose of relational and hypersafety verification? 
Any specific choice of a product program can be viewed as a {\em reduction} of the set of all alignments (or the behaviours of the parallel product). For a  group of $k$ syntactic runs, one from each individual component, a (minimal) {\em reduction} includes precisely one alignment of them, and discards all others. For recursive programs, with syntactic runs of unbounded length, there are infinitely many choices that can be made and as such infinitely many reductions that can be defined.


Our key (theoretical) contribution in this paper is to define an infinite space of product programs (reductions). Each reduction is representable as standard recursive programs and therefore can be proved correct using standard  verification techniques. Each member of this space is then formally definable and algorithmically constructible, and the set is {\em countable} and (recursively) {\em enumerable}, in the same style as \cite{FarzanV2019,FarzanV2020}.

To define the state space of these programs, we rely on the observation that the task of choosing one alignment from a group is like {\em scheduling}. Intuitively, one can think of the resolution of the nondeterminism to be done through the choices made by a {\em scheduler}. Most importantly, we observe that these {\em schedulers} can also be modelled in the same formalism as recursive programs. Given a set of statements, coming from distinct components, one can formally define a set of {\em schedules} (which are basically sequences over these statements), which can then serve as a restriction of the set of all alignments. For instance, a lockstep scheduler of two (distinct) components can be viewed as one that admits any schedule {\em alternating} statements from the two components. Construction of a reduction then becomes a standard {\em intersection} of the set of schedules defined by the scheduler and the well-nested shuffle of the components; a construction that is supported by the theory of visibly pushdown languages for both automata and grammars. In Section \ref{sec:reduction}, we formalize these schedulers as {\em contextual lexicographical orders} and the yielded reductions as {\em contextual lexicographical reductions}, {\em lex reductions} for short. Here, {\em context} refers to the fact that the schedule is {\em stateful} and adjusts its future choices based on what it has done in the past. These can be viewed as generalization of similar concepts from \cite{FarzanV2019,FarzanV2020,FarzanKP2022} to recursive programs.

We use ideas from {\em trace theory} \cite{Mazurkiewicz1987} (a simple theory for concurrency) to argue that our reductions are constructible. However, in Section \ref{sec:constructing-reductions}, we propose a new optimized algorithm for constructing a representation of our contextual lexicographic reductions as {\em visibly pushdown automata and grammars}. The new construction leverages the fact that, in the context of hypersafety (and relational) verification, we have the {\em full commutativity} of the actions from different components, to propose a {\em simpler} and {\em more compact} construction of visibly pushdown (VP) lex reductions. 

\begin{wrapfigure}[8]{r}{0.5\textwidth}\iftoggle{ext}{}{\vspace{-12pt}}
\includegraphics[scale=0.6]{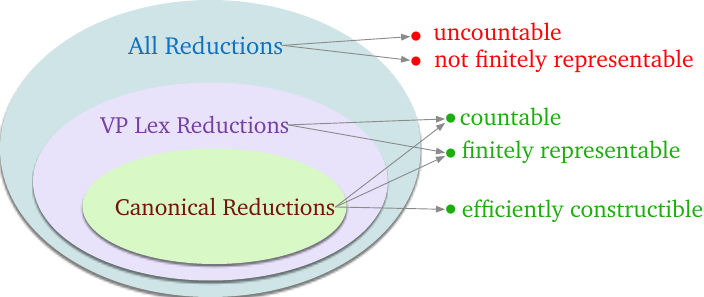}\vspace{-10pt}
\end{wrapfigure}
The specific feature of our proposed space of reductions is that, unlike the set of all reductions, it is recursively enumerable.
This means that assuming the existence of a verifiable reduction in the space, one can set up a trivial decision procedure to find it.
The underlying verification problem remains undecidable for generic infinite-state programs, but at least, we know that the component of searching for a reduction does not add any new divergence opportunities. This is in sharp contrast to, for example, considering the state space of {\em all possible} reductions, where one has to either bound the space to a finite subset \cite{BartheCK2011} or use ad-hoc heuristics \cite{MordvinovF2017,SousaD2016} to curb this intrinsic divergent behaviour. Moreover, the VP lex reductions are finitely representable, whereas an arbitrary reduction may not be.

Rather than setting up a search algorithm for a reduction in the space of VP lex reductions, we opted for investigating a different more practical hypothesis in this paper. This is motivated by the fact that the search, even in the simpler setup of  \cite{FarzanV2019,FarzanV2020}, is very expensive and unscalable and the problem only gets worse with the increased complexity of the same operations for pushdown systems \cite{EsparzaHRS2000}.
So, there is a natural question about what to do when one hits that scalability barrier. There is a well-established tradition of user-provided annotations to help with the automation of harder instances. For instance, in \cite{ShemerGSV2019}, when the mined predicates are insufficient for the combined search for an alignment and a proof, the user can input predicates in the proof (of the unknown alignment) to help; which is in line with this tradition.   
In this paper, we explore the alternative possibility of user-provided annotations for the specification of reductions. we devise a smaller (yet infinite) subfamily of lex reductions, we call {\em canonical reductions}, that have three properties: 
\begin{enumerate} 
\item They are intuitive to the programmer. For instance, looking at $\mathtt{div}$ and \ref{p2}, the programmer can easily guess and specify that they are interested in an asynchronous product of two $\mathtt{div}$ instances where one copy moves twice as fast as the other copy.
\item They are simply and succinctly specifiable (like a correctness property is in a logical language) as part of the verification command, using {\em composable building blocks}.
\item They are expressive enough to provide the means of solving a broad range of instances.
\end{enumerate}


To understand the obstacles in achieving this, let us expand our simple $\mathtt{div}$ example one more time with a substantially more sophisticated property:
\begin{align}
n'' = n + n' &\implies \mathtt{div}(n'',d) \ge \mathtt{div}(n,d) + \mathtt{div}(n',d) \label{p3} \tag{\sc distributivity} 
\end{align}
where \ref{p3} is a 3-safety property, and clearly, {\em lockstep} product is not the right choice for this and varying the rate is also insufficient. Consider the product program 
$$\left(\mathtt{div(n,d)} \oplus \mathtt{div(n',d')}\right) \otimes \mathtt{div(n'',d'')}$$ 
 where we intentionally use the vague $\oplus$ and $\otimes$ operators for the choice of product operators for now. 
Intuitively, since $n'' = n + n'$, we expect the execution length of the $\mathtt{div(n'',d'')}$ copy to be roughly the same as the sum of the execution lengths of the $\mathtt{div(n,d)}$ and $\mathtt{div(n',d')}$ copies. Therefore, in a suitable product program, we want to run a kind of {\em sequential composition} of $\mathtt{div(n,d)}$ and $\mathtt{div(n',d')}$ in tandem with $\mathtt{div(n'',d'')}$; that is, align each of the first $n$ recursive calls of $\mathtt{div(n'',d'')}$ with a call from $\mathtt{div(n,d)}$ and the last $n'$ recursive calls of $\mathtt{div(n'',d'')}$ with a call from $\mathtt{div(n',d')}$. This is precisely how looped computations are split \cite{BartheCK2011} to accommodate the construction of an effective product program. However, a recursive computation cannot be split like a loop through  {\em loop fission}, by breaking a loop  into two sequentially composed loops. The solution, at the high level, is in fact the complete opposite. One first {\em merges} $\mathtt{div(n,d)}$ and $\mathtt{div(n',d')}$ so that it looks like one recursive sequence of calls, and then synchronizes this merged program with $\mathtt{div(n'',d'')}$ in a lockstep manner. This {\em merge} can be viewed as a recursive analog of {\em loop fusion}.

\begin{wrapfigure}[4]{r}{0.42\textwidth}\vspace{-10pt}
\includegraphics[scale=0.5]{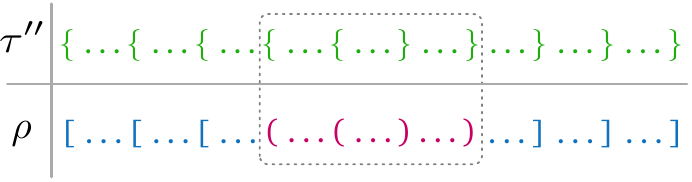}\vspace{-10pt}
\end{wrapfigure}
To visualize this, recall the runs  $\tau$ and $\tau'$ illustrated earlier, and consider $\tau''$ illustrated on the right, as being respectively runs of the three disjoint copies $\mathtt{div(n,d)}$, $\mathtt{div(n',d')}$, and $\mathtt{div(n'',d'')}$.  First, $\tau$ and $\tau'$ are {\em merged} together, by {\em fully nesting} $\tau'$ inside $\tau$ to get $\rho$. This operation (standing in for $\oplus$) can be viewed as the recursive analog of {\em concatenation} in the looped (iterative) setting.  

We introduce {\em nested concatenation}, the recursive dual of concatenation as a {\em novel building block}. Note that,  for  a {\em tail recursion}, nested concatenation behaves similar to concatenation, but otherwise, it is an entirely different and new way of composing computations. $\rho$ can then be composed with $\tau''$ through a standard lockstep composition, where the green and blue/pink calls are executed in tandem. Effectively,  the innermost two calls of $\tau''$ are synched with the two calls of $\tau'$, and its outermost three calls are synched with the calls of $\tau$. 

Looking at the \ref{p1} and \ref{p2} properties of $\mathtt{div}$, a {\em parametric-by-speed} version of lockstep is also a suitable candidate for another building block. Note that the parametric lockstep, by itself, yields an unbounded family of reductions. Combined with nested concatenation, they give rise to an infinite number of possibilities for product programs that can be very succinctly specified by the user. These building blocks are highly intuitive, but very nontrivial to define formally and construct correctly. Our second (theoretical and practical) key contribution is that we propose formal definitions for these building block reductions (Section \ref{sec:cr}), and also show that these reductions all belong to the state space of lex reductions. Moreover, we propose an alternative direct way of constructing product programs made from the two building block reductions of {\em nested concatenation} and {\em parametrized lockstep} (Section \ref{sec:eccr}) that yields much more compact representations for these canonical reductions, compared to the generic (but optimized) construction from Section \ref{sec:constructing-reductions} for all lex reductions. Remarkably, the subfamily demonstrates a large degree of expressivity with nearly all of our 103 benchmarks becoming solvable using a a reduction from this family.

As a bonus (third) contribution, we make a new observation for the verification of concurrent recursive programs: The well-nested shuffle is an under-approximation of a {\em concurrent recursive program}, which itself is a standard (nondeterministic) recursive program. As such, any existing technique for recursive program analysis and bug finding can be used to find errors in it which are guaranteed to be errors in the original program. Moreover, in Section \ref{sec:concurrency}, we formulate conditions under which this under-approximation becomes a {\em sound} one; that is, any property proven for the well-nested shuffle (using standard techniques) is guaranteed to hold for the original program. 

In Section \ref{sec:semantics}, we outline how the language-theoretic (automata or grammar) description of a product program can be turned into a classic recursive program and encoded for standard safety verification, which serves as the last step of our reduction of the problem of hypersafety verification of recursive programs to standard safety verification.

We have a tool called \tool, that reduces an input hypersafety instance to a classic safety instance by constructing the appropriate product programs. We experiment with the a large set (103) of benchmarks to show that the product programs are absolutely vital for hypersafety verification of recursive programs.  Specifically, our experiments with 103 benchmarks confirm that the high level idea of the product program is intuitively known to the programmer and can be succinctly specified and facilitates the verification of the hypersafety instance.

\section{Background}\label{sec:background}
\subsection{Languages, Automata, and Grammars} \label{sec:lag}
\emph{Visibly pushdown langauges} \cite{AlurM2004} are a well-established formal model for capturing the set of \emph{interprocedurally realizable paths} \cite{RepsHS1995,Ramalingam2000} of a program: procedure calls always return to the correct location, but the (data) feasibility of the paths is not considered.

We review some key concepts for visibly pushdown languages, and refer the reader to \cite{AlurM2004} for a full exposition. 
A \emph{visibly pushdown alphabet} (VP alphabet) $\tS$ is the union of three disjoint finite alphabets, $\Sc$ (\emph{calls}), $\Sr$ (\emph{returns}), and $\Sl$ (\emph{internals}). We use $\tS$ when we want to emphasize the alphabet is visibly pushdown and use $\Sigma$ for a general alphabet. A word $w_1\cdots w_l \in \tS^*$ induces a unique \emph{matching relation} \cite{AlurM2009} $\rightsquigarrow$ over $\{-\infty, 1, 2, \ldots, l\} \times \{1, 2, \ldots, l, +\infty\}$ so that $i \rightsquigarrow j$ ($-\infty < i < j < +\infty$) means $w_i, w_j$ are a pair of call and return that match in the usual sense, and $i \rightsquigarrow +\infty$ (resp. $-\infty \rightsquigarrow i$) means $i$ is a \emph{pending} call (resp. return). A word without pending letter is \emph{well-matched}.

A \emph{visibly pushdown automaton} (VPA) \cite{AlurM2004} over a VP alphabet is a pushdown automaton whose stack operations are entirely controlled by the type of the letter read: it pushes when reading a call, pops when reading a return, and neither pushes nor pops when reading an internal.

\begin{definition}[Visibly Pushdown Automaton (VPA)]
  A visibly pushdown automaton over a VP alphabet $\tS$ is a tuple $M = (Q, Q^{\mathsf{in}}, \Gamma, \delta, Q^{\mathsf{F}})$ where $Q$ is a finite set of states, $Q^{\mathsf{in}} \subset Q$ is a set of initial states, $\Gamma$ is a finite stack alphabet that contains a special bottom-of-stack symbol $\bot$, $\delta = \dc \cup \dr \cup \dl$ is a set of transitions, where $\dc \subset Q\times\Sc\times Q\times(\Gamma\setminus\{\bot\})$, $\dr \subset Q\times\Sr\times\Gamma\times Q$, and $\dl \subset Q\times\Sl\times Q$, and $Q^{\mathsf{F}} \subset Q$ is a set of final states.
\end{definition}

We say a VPA is complete if at any state, reading any letter leads to a next state.
A \emph{visibly pushdown language} (VPL) is a language recognized by some VPA. VPLs can also be characterized by \emph{visibly pushdown grammars} (VPGs) \cite{AlurM2004}, a subclass of context-free grammars. We are mainly concerned with languages of well-matched words, which are generated by \emph{well-matched VPGs} \cite{JiaKT2021}. We additionally permit a grammar to have more than one start symbols.

\begin{definition}[Well-matched VPG]
  A context-free grammar $G = (V, \tS, P, S)$ is a \emph{well-matched VPG} if every production rule in $P$ has the form $X \to \epsilon$, $X \to a Y$, or  $X \to c Y r Z$ 
   where $a \in \Sl$, $c \in \Sc$, and $r \in \Sr$.
\end{definition}

We recall the {\em shuffle} operator $\shuffle$ over languages. For any alphabet $\Sigma$ and $\Sigma' \subset \Sigma$, let the projection function $\Pi_{\Sigma'}:\Sigma^{*}\to\Sigma'^{*}$ be the string homomorphism that erases letters that are not in $\Sigma'$.
\begin{definition}[Shuffle]
    Let $\Sigma$ be an alphabet that can be partitioned into a disjoint union $\Sigma_{1}\uplus\cdots\uplus\Sigma_{n}$. Let $L_{i}$ be a language over alphabet $\Sigma_{i}$ ($i=1,\ldots,n$). The \emph{shuffle} $L_{1}\shuffle\cdots\shuffle L_{n}$ is the set of words $w\in\Sigma^{*}$ such that $\Pi_{\Sigma_{i}}(w)\in L_{i}$ for each $i$.
\end{definition}
Regular languages are closed under shuffle, but this is not true for VPLs, even in the case where $L_{i}$'s are copies of the same language. For instance, the shuffle of two Dyck languages \cite{DavisSW1994} over disjoint VP alphabets (opening/closing delimiters are calls/returns) is not context-free \cite{HopcroftMU2007}. 

\subsection{Recursive Programs}\label{sec:rp}
Let $V$ be the set of all valuations, i.e., maps from program variables to values, and $\stack(V)$ be the set of all nonempty stacks whose frames are elements of $V$. We use $S.\nu$ to denote a stack whose top frame is $\nu$ and the rest is $S$, which is a smaller stack.

Following \cite{HeizmannHP2010}, we model a recursive program as a VPL $P$ over a VP alphabet $\tS$ and a semantics $\mathcal{S}: \tS \to \powerset(\stack(V) \times \stack(V))$ which respects the VP alphabet $\tS$: the type of stack operation performed by $\sem{S}{a}$ is determined by the type of $a$. For our purpose, we assume $P$ is well-matched. Formally, there is a map $\mathcal{F} \subset (\Sc \times V \times V) \cup (\Sr \times V \times V \times V) \cup (\Sl \times V \times V)$ such that
\begin{itemize}
  \item If $c \in \Sc$, then $\sem{S}{c} := \{ (S.\nu,\, S.\nu.\nu') : \nu' \in \sem{F}{c}(\nu) \}$. 
  \item If $r \in \Sr$, then $\sem{S}{r} := \{ (S.\nu_{<}.\nu,\, S.\nu') : \nu' \in \sem{F}{r}(\nu,\nu_{<}) \}$.
  \item If $a \in \Sl$, then $\sem{S}{a} := \{ (S.\nu,\, S.\nu') : \nu' \in \sem{F}{a}(\nu) \}$.
\end{itemize}
Extend $\mathcal{S}$ from letters to syntactic runs (words over $\tS$ without pending return) via relation composition. For a syntactic run $\rho$, the hoare triple $\{\pre\} \rho \{\post\}$ holds if for all $(S.\nu,\, S'.\nu') \in \sem{S}{\rho}$, if $\nu \models \pre$, then $\nu' \models \post$. We say $\rho$ is \emph{feasible} if $\sem{S}{\rho} \neq \varnothing$. If $\rho$ is infeasible, any triple $\{\pre\} \rho \{\post\}$ is vacuously true.
Given a recursive program $(P, \mathcal{S})$, define $\{\pre\} P \{\post\}$ as $\forall \rho \in P.\, \{\pre\} \rho \{\post\}$. When $\mathcal{S}$ is clear from the context, we sometimes use just $P$ to refer to the recursive program.

\subsection{Commutativity and Reductions}

\begin{definition}[Commutativity-based Equivalence \cite{Mazurkiewicz1987}]\label{def:ceq}
Given an irreflexive symmetric \emph{commutativity relation} $I \subset \Sigma \times \Sigma$, define the \emph{commutativity-based equivalence} $\equiv_I$ on $\Sigma^*$ as the least transitive reflexive relation such that for all $u, v \in \Sigma^*$ and $a, b \in \Sigma$, if $(a, b) \in I$, then $uabv \equiv_I ubav$.
\end{definition}

It is straightforward to verify that $\equiv_I$ is indeed an equivalence relation. 
Intuitively, two words $u, v$ are equivalent under $\equiv_I$ if $v$ can be obtained from $u$ by repeatedly swapping adjacent commuting letters. The equivalence relation $\equiv_I$ partitions a language into equivalence classes. 
If a subset of the language contains representatives from all equivalence classes, we call it a \emph{reduction}.

\begin{definition}[Reduction]\label{def:reduction}
    Given a commutativity relation $I$ on $\Sigma$, for a language $L$ over $\Sigma$, a subset $\widehat{L}$ of $L$ is said to be a \emph{reduction} if for all $\rho \in L$, there exists $\sigma \in \widehat{L}$ such that $\sigma \equiv_I \rho$.
    A reduction is \emph{minimal} if there is precisely one such $\sigma$ for every $\rho$.
\end{definition}

In our context, letters are statements and words are runs of a program. Thus, we are interested in not an arbitrary commutativity relation but one that is \emph{sound}: if $(a, b) \in I$, then $ab$ and $ba$ have the same semantics. If $\widehat{P}$ is a reduction of $P$ under a sound commutativity relation, then $\widehat{P}$ covers all semantics of $P$, which results in the following theorem:

\begin{theorem}[Soundness of Reductions]\label{thm:soundness-of-reductions}
    If $\widehat{P}$ is a reduction of $P$ under a sound commutativity relation, then
    $\{\pre\}\ \nP\ \{\post\}$ if and only if $\{\pre\}\ P\ \{\post\}$.
\end{theorem}


\section{Reductions Space for Hypersafety Verification } \label{sec:reduction}
The classic way \cite{BartheDR2004} to verify a hypersafety property for a recursive procedure is to produce $n$ disjoint copies $P_1, \dots, P_n$ from it and state the pre/postconditions $\pre$ and $\post$ as a property of the sequential composition of these disjoint copies:
$\{\pre\}\ P_1 ; \dots ; P_n\  \{\post\}$. Observe that the same disjointness assumption implies that the sequential composition can be replaced by a parallel composition (or the shuffle). Hence we can state our goal as constructing a single (product) recursive program $\nP$ such that for any hypersafety property expressed as a pair of pre/postconditions $\pre$ and $\post$:
\[\{\pre\}\ P_1 \shuffle \dots \shuffle P_n\ \{\post\} \iff \{\pre\}\ \nP\ \{\post\} \tag{Soundness}\label{con:sound}\] 
which makes it sound to verify $\nP$ instead. The parallel product is more useful for our purpose, because it already includes every possible imaginable alignment of individual program runs inside it. We start by assuming that the programs $P_i$'s are entirely disjoint, and later (in Section \ref{sec:concurrency}) address the concurrency case where some may share variables.

Each $P_i$ is a classic recursive program over a disjoint memory and, in the classic sense, can be interpreted using a stack semantics in the style outlined in \ref{sec:rp}. In order to track the semantic meaning of an arbitrary run of a {\em product} of these programs, which may align the executions from the programs in an arbitrary way, the semantic function has to maintain $n$ stacks, one for each component, to keep track of each of the $n$ disjoint recursive computations. In other words, the global state is now a {\em tuple} of $n$ stacks, and a statement from the $i$-th component operates on the $i$-th stack from this tuple. To get a classic recursive program, we need this tuple to collapse into a single stack while preserving the ability to keep track of the computation of any arbitrary component $P_i$. This can be guaranteed precisely when the alignments are \emph{well-nested}. 
\begin{definition}[Well-nested Words] \label{def:wnw}
Let $\tS = \tS_1 \uplus \cdots \uplus \tS_n$, with disjoint $\tS_i$'s. A word $w \in \tS^*$ is \emph{well-nested} if $i \rightsquigarrow j$ (where $\rightsquigarrow$ is the matching relation on $w$) implies that $w_i$ and $w_j$ both belong to the same $\tS_k$. For any language $L \subset \tS^*$, $\wn(L)$ denotes the set of all well-nested words in $L$. 
\end{definition}
If a language over $\tS$ contains only well-nested words, we say the language is well-nested. Note that being well-nested is different from being well-matched (see Section \ref{sec:lag}).

We need a suitable representation for individual $P_i$'s and $\nP$ that facilitates this goal. The control flow graph (or equivalent) representations lose track of the recursive structure of the runs. Context free grammars (CFG) or a pushdown automaton (PA) are the classic models that maintain this structure. However, 
{\em visibly pushdown languages} (VPL) \cite{AlurM2004}, or alternatively the {\em regular nested word languages} (RNWL) \cite{AlurM2009}, are a strict subset of context free languages that are expressive enough to model recursive programs and yet have the same desirable boolean properties as regular languages.

It is straightforward to take as input the code of a recursive component $P_i$ and construct a {\em visibly pushdown grammar} (VPG) or {\em visibly pushdown automaton} (VPA) that precisely captures the set of well-matched syntactic runs of $P_i$.  We abuse the notation $P_i$ to refer both to the program and the set of syntactic runs, which form a well-matched VPL, since the distinction is always clear from context. We can always ensure that each $P_i$ is over a distinct visibly pushdown alphabet $\tS_i$, and introduce $\tS = \tS_1 \uplus \cdots \uplus \tS_n$.

The language $P_{\shuffle} = P_1 \shuffle \dots P_n$ is not necessary visibly pushdown, even if all $P_i$'s are; nor does it exclude ill-nested alignments.
In this section, we formulate how we can use simple observations from language theory and concurrency theory to define and algorithmically construct individual product programs (reductions) from $P_{\shuffle}$, which (1) satisfy the \ref{con:sound} condition, (2) only contain well-nested alignments, and (3) are visibly pushdown and therefore representable through VPAs and VPGs. It is very important to understand that we manipulate the language (i.e. the set of syntactic runs) to achieve this, and {\em do not} formulate this as a manipulation of any specific representation of it. So, even though we operate on sets of syntactic runs, we do not manipulate (bounded) syntactic structures (grammars, automata, control flow graphs, etc) to achieve this. Syntactic structures are bounded and assume bounded amount of manipulation; that is, unless one resorts to unfolding and therefore changing the structures. Languages are unbounded and as such present an unbounded amount of potential for selection of product programs.

\def\Ith{\mathbb{I}}
\subsection{Contextual Lexicographic Reductions} \label{sec:contextual-lex-reductions}
For every run $\rho \in P_{\shuffle}$, it suffices that the reduction $\nP$ contains {\em at least one} permutation of statements of $\rho$ while maintaining the order of the statements within each $P_i$. All other permutations are {\em equivalent} and it suffices that the entire {\em equivalence class} (see Definition \ref{def:ceq}) has at least one of its elements. This is a {\em commutativity} type argument, where pairs of statements from different copies can be considered to soundly commute, since they operate on disjoint memories. If the goal is to opt for programs with simpler proofs then we select {\em precisely one} representative, since it is generally easier to prove a program with fewer behaviours correct. Such a language $\nP$ is by definition a \emph{minimal reduction} (see Section \ref{sec:background}) of $P_{\shuffle}$ under the commutativity relation $\Ith = \{ (a,b) \mid a \in \tS_i, b \in \tS_j, i \neq j \}$, the largest sound commutativity relation. 

Note that for any program that is not bounded, there are infinitely many equivalence classes of $\Ith$, and as long as these equivalence classes have more than one member, then \emph{uncountably} many minimal reductions (see Definition \ref{def:reduction}), each uniquely determined by a {\em selector} of representatives from the equivalence classes under $I$.

\begin{proposition}\label{prop:uncountable}
In general, the set of {\em minimal reductions} of $P_{\shuffle}$ under $\Ith$ is not {\em countable}.  
\end{proposition}

Our goal is to introduce an infinite yet finitely representable (and thus countable) subset of this. To this end, we are interested in a selector (from equivalence classes) that can be finitely represented. Inspired by \cite{FarzanV2020,FarzanKP2022}, we focus on selectors of the form $E \mapsto \min_{\preceq}(E)$, where $E$ is an equivalence class, and $\preceq$ is a \emph{contextual lexicographic order}, a generalization of the standard lexicographic orders. We define a \emph{contextual order} on $\tS$ as a map from $\tS^*$ to strict partial orders on $\tS$ (denoted $\PO(\tS)$);
this allows the order on letters to vary according to the left context, which is a word over $\tS$.
Each contextual order on $\tS$ determines a contextual lexicographic order on $\tS^*$.

\begin{definition}[Contextual Lexicographic Order] \label{def:clo}
    Given a contextual order $\prec:\tS^*\to\PO(\tS)$, we lift $\prec$ to a \emph{contextual lexicographic order} (CLO) on $\tS^*$ by having $\sigma \preceq \rho$ if and only if
    \begin{itemize}
        \item $\sigma$ is a prefix of $\rho$, or
        \item $\sigma = \alpha a \beta$ and $\rho = \alpha b \gamma$ for some $\alpha, \beta, \gamma \in \tS^*$ and $a,b\in\tS$ such
        that $a\prec_{\alpha}b$.
    \end{itemize}
\end{definition}

The standard lexicographic orders are the special case where the map $\prec$ is a constant function and its image is a total order.

\begin{proposition}\label{prop:total-clo}
    If $\prec$ is a contextual order on $\tS$, then the CLO induced by $\prec$ is a partial order on $\tS^*$. Furthermore, if $\prec_x$ is a total order on $\tS$ for every $x \in \tS^*$, then the CLO induced by $\prec$ is total.
\end{proposition}

\begin{example}[A Round-Robin\footnote{Round-robin here is an intuitive simplification of (1,1)-lockstep, whose formal definition is presented in Section \ref{sec:cr}.} CLO]\label{ex:rr}
    Let $\tS_1 = \{\texttt{(} ,\texttt{)}\}$ and $\tS_2 = \{\texttt{[} , \texttt{]}\}$ be the VP alphabet of parentheses and brackets respectively, and let $\tS = \tS_1 \uplus \tS_2$. Define the contextual order $\prec$ as follows: for any words $x, w \in \tS^*$ such that $w$ is nonempty and well-matched, 
    \begin{align*}
        x \texttt{(} &\mapsto \texttt{[} < \texttt{(} < \texttt{)} < \texttt{]} \\
        x \texttt{(} w &\mapsto \texttt{)} < \texttt{(} < \texttt{[} < \texttt{]} \\
        x \texttt{[} w &\mapsto \texttt{]} < \texttt{[} < \texttt{(} < \texttt{)} \\
        \text{otherwise} &\mapsto \texttt{(} < \texttt{[} < \texttt{]} < \texttt{)}
    \end{align*}
    Intuitively, $\prec$ simulates round-robin scheduling and it enforces well-nestedness. Let $\preceq$ be the CLO induced by the contextual order $\prec$. Consider $\alpha = \texttt{([()])}$, $\beta = \texttt{(([]))}$, and $\gamma = \texttt{([(]))}$. They are in the same equivalence class under $\Ith$, and we have $\alpha \preceq \beta$ and $\alpha \preceq \gamma$. Note that $\beta$ does not follow a round-robin scheduling and $\gamma$ is ill-nested, and $\alpha$ is selected over both by this CLO.
    \qed
\end{example}

If $\preceq$ is a CLO, the set of $\preceq$-minimal representatives of the equivalence classes, induced by a commutativity relation $I$, of a language $L$ form a \emph{lex reduction} of $L$.

\begin{definition}[Lex Reduction]\label{def:lexreduction}
    Let $L$ be a language over $\tS$ and $I \subseteq \tS \times \tS$ a commutativity relation. The reduction of $L$ induced by contextual order $\prec$ on $\tS$ is defined as
    \[
    \red_{I,\prec}(L):=\{w \in L \mid \forall u \in L:\ u \equiv_I w \land u \preceq w \implies u = w \},
    \]
    where $\preceq$ is the CLO induced by $\prec$.
\end{definition}

If $\preceq$ is total, then $\red_{I,\prec}(L)$ is a minimal reduction. For simplicity, we assume from now on that the contextual orders in consideration map words over $\tS$ to strict total orders on $\tS$ (denoted $\Lin(\tS)$), which is a sufficient (but not necessary) condition for the order yielding a minimal reduction. Recall that $\Ith$ is the relation that declares any letters from distinct alphabets as commuting. Since we are strictly interested in this commutativity relation, we use the shorthand notation $\red_\prec$ in place of $\red_{\Ith,\prec}$, since the commutativity relation will always be $\Ith$.

\begin{example} \label{exm:lex-reductions}
    We continue in the setting of Example \ref{ex:rr}. Observe that for \emph{any} $P_1 \subset \{ \texttt{(}^n \texttt{)}^n : n > 0 \}$ and $P_2 \subset \{ \texttt{[}^n \texttt{]}^n : n > 0 \}$, the reduction $\red_\prec(P_1 \shuffle P_2)$ is precisely the round-robin like scheduling of $P_1$ and $P_2$ that also enforces well-nestedness.
\end{example}


If $I$ is a sound commutativity relation (such as $\Ith$ in the hypersafety context), then any reduction under $I$ can be verified in place of the original program (Theorem \ref{thm:soundness-of-reductions}). This naturally applies to lex reductions, which we restate as follows:

\begin{theorem}[Soundess of Lex Reductions]\label{thm:soundness-of-lexreduction}
Let $I$ be a sound commutativity relation over $\tS^*$ and $\prec$ a contextual order on $\tS$. The lex reduction $\red_{I,\prec}(P_1 \shuffle \dots \shuffle P_n)$ satisfies the \ref{con:sound} condition.
\end{theorem}



We conclude our discussion around lex reductions with the following observations, which motivate the content of the next section. 

\begin{proposition}\label{prop:nonvpllex}
    There exists a contextual order $\prec$ such that $\red_\prec(P_1 \shuffle P_2)$ is not context-free. There also exists a contextual order $\prec$ such that $\red_\prec(P_1 \shuffle P_2)$ contains ill-nested words.
\end{proposition}


\subsection{Visibly Pushdown Well-Nested Lex Reductions}
The conclusion from the results presented in the previous section is that we need to further restrict the set of lex reductions so that we can get reductions that are well-nested and visibly pushdown. We achieve this by restricting the set of contextual orders in consideration.

\begin{definition}[Visibly Pushdown Contextual Order (VPO)]\label{def:representability}
A contextual order $\prec: \tS^{*} \to \Lin(\tS)$ is \emph{represented by a deterministic VPA} $A = (Q_A,Q_A^{\mathsf{in}},\Gamma_A,\delta_A,Q_A^{\mathsf{F}})$ if $A$ is complete and there is a map $\ord: Q_A \to \Lin(\tS)$ such that for any $w \in \tS^*$, if the run of $A$ on $w$ ends at state $q$, then ${\prec_w}$ is $\ord(q)$. We call such orders {\em visibly pushdown contextual orders}.
\end{definition}
An alternative way of interpreting this definition, for the reader unfamiliar with language theory, is to think about the contextual order to be definable through a scheduler that has the same computation power as a visibly pushdown automaton.
\begin{example}
    The contextual order $\prec$ in Example \ref{ex:rr} is visibly pushdown, which can be seen as follows. 
    Under the map $\prec$, the preimage of each order on $\tS$ is a VPL: the four different types of prefixes (i.e., context) that would lead to the four different choices of orders are all recognizable by deterministic VPA. Then, the product of the four VPAs (as in the proof of closure under intersection \cite{AlurM2004}) can represent the contextual order $\prec$.
\end{example}

The set of all minimal runs defined by any visibly pushdown order is visibly pushdown. In other words, if the order is VP, then all schedules (independent of the choice of language) that are accepted by the order form a VP language.
\begin{proposition}\label{prop:vp-sigma-star}
    If $\prec$ is a VPO, then $\red_\prec(\tS^*)$ is visibly pushdown.
\end{proposition}

Yet, this result is not strong enough to imply that a minimal reduction of $P_1 \shuffle \cdots \shuffle P_n$ induced by a VPO is visibly pushdown or well-nested.

\begin{proposition}\label{prop:dummy}
There exists a VPO $\prec$ such that $\red_\prec(P_1 \shuffle P_2)$ is not context-free. There also exists a VPO $\prec$ such that $\red_\prec(P_1 \shuffle P_2)$ contains ill-nested words.
\end{proposition}

What is the missing piece? Interestingly, there is a correlation between the reduction being well-nested and it being visibly pushdown. Namely, if we enforce well-nestedness, we get visibly pushdown for free.

\begin{theorem}[Visibly Pushdown Lex Reduction]\label{thm:vpwnred}
    If $\prec$ is visibly pushdown and $\red_\prec(P_1 \shuffle P_2) \subset \wn(P_1 \shuffle P_2)$, then $\red_\prec(P_1 \shuffle P_2)$ is visibly pushdown.\footnote{For notational simplicity, the statement is made only for two components, but this theorem and the ones below also hold for $n > 2$ components.}
\end{theorem}

The problem is that the above condition is dependent on the specific choices of $P_i$'s and is rather non-constructive (or equivalently hard to check). 
To obtain a condition on the order alone, independent of $P_i$'s, we introduce the concept of \emph{coherent contextual orders}. Intuitively, a coherent contextual order deprioritizes any returns that do not match the last open/pending call, because if the non-matching return is prioritized, then one ends up violating well-nestedness. Formally:

\begin{definition}[Coherent Contextual Order] \label{def:coherence}
   Let $\tS_i$ be VP alphabets ($i = 1,2$) and $\tS = \tS_1 \uplus \tS_2$. A contextual order $\prec$ on $\tS$ is said to be {\em coherent} if for any word $u \in \tS^*$ and $i \in \{1,2\}$, if $u$ has pending calls and the last one is from $\Sc_i$, then for any letter $a \in \tS_i$ and $r \in \bigcup_{j \neq i} \Sr_j$, we have $a \prec_u r$.
\end{definition}

\begin{example}
    The contextual order $\prec$ in Example \ref{ex:rr} is coherent. For instance, if $u = x \texttt{(} w$, then the last pending call in $u$ is $\texttt{(} \in \Sc_1$, and we can verify that $a \prec_u \texttt{]}$ for all $a \in \tS_1$.
\end{example}

For specific languages $P_1, P_2$, coherency is not a necessary condition for the well-nestedness of $\red_\prec(P_1 \shuffle P_2)$. For instance, the word $u$ that witnesses the violation of coherency may not occur as a prefix in the shuffle $P_1 \shuffle P_2$ at all. However, coherency is always a sufficient condition.

\begin{proposition}\label{prop:coherency-implies-well-nestedness}
    Let $\prec$ be a coherent contexual order on $\tS = \tS_1 \uplus \tS_2$ and $P_1, P_2$ be well-matched VPLs over $\tS_1, \tS_2$ respectively. Then $\red_\prec(P_1 \shuffle P_2) \subset \wn(P_1 \shuffle P_2)$.
\end{proposition}

Note that a coherent contextual order is not necessarily visibly pushdown, and vice versa. As a consequence of Theorem \ref{thm:vpwnred} and Proposition \ref{prop:coherency-implies-well-nestedness}, we can conclude that the combination of the two properties in an order is sufficient to guarantee that it induces a reduction that is both visibly pushdown and well-nested.

\begin{theorem}[VP and Coherent Contextual Order]\label{thm:vpred-vp-coherent}
    If $\prec$ is a visibly pushdown {\bf and} coherent contextual order, then $\red_\prec(P_1 \shuffle P_2)$ is visibly pushdown and well-nested.
\end{theorem}

Once we know a reduction is visibly pushdown, we know that it is theoretically constructible. However, in the next section, we take a closer look at what is the best way of algorithmically constructing a visibly pushdown and well-nested $\red_\prec(P_1 \shuffle P_2)$.

\section{Constructing Reductions}\label{sec:constructing-reductions}
A language $L$ over $\tS$ is called \emph{closed} under a commutativity relation $I$ if each equivalence class of $I$ is either entirely included in $L$ or none of the members appear in $L$.
For example, $P_1 \shuffle \dots \shuffle P_n$ is by definition closed under any commutativity relation that does not reorder statements from the same $P_i$. 
It is straightforward to see that if $L$ is closed under $I$ then  
\begin{equation}
\red_{I,\prec}(L)=\red_{I,\prec}(\tS^{*})\cap L. \label{eq:lex}
\end{equation} 
This simple observation puts forward an algorithmic path for the construction of the reductions in the iterative case \cite{FarzanV2019,FarzanV2020,FarzanKP2022}. The generic language $\red_{I,\prec}(\tS^{*})$, which is somewhat independent of any given verification instance and is purely defined based on an alphabet of statements, can be constructed by an adaptation of the idea of {\em sleep sets} \cite{Godefroid1996}. Then, any reduction can be constructed through a simple construction of the intersection of this generic language and the baseline parallel product $P_1 \shuffle \dots \shuffle P_n$. 
The proof of Theorem \ref{thm:vpred-vp-coherent} puts forward a similar type of construction for the visibly pushdown case. However, this is unsatisfactory in two orthogonal ways:
\begin{enumerate}
\item {\em Complexity:} the construction is generic and does not exploit properties of the commutativity relation. We are primarily interested in the commutativity relation $\Ith$, and as we demonstrate in Section \ref{sec:oc}, there is a construction that is strictly simpler under the assumption that all copies fully commute and the stack alphabet is small.
\item {\em Burden on user:} Theorem \ref{thm:vpred-vp-coherent} demands two independent conditions on the order. The visibly pushdown condition on the order is rather trivial; if the order can be expressed in a specific template (grammar or automata), then it is visibly pushdown. However, the coherence condition is an extra burden to check/verify, which should ideally be avoidable.
\end{enumerate}
Here, we give an alternative construction of VP lex reductions that overcomes these restrictions. 

\subsection{Well-nested Shuffle}
Recall that the role of the {\em coherence} condition is to ensure that the reduction only includes well-nested alignments, or in other words, it is a subset of $\wn(P_1 \shuffle \dots \shuffle P_n)$. It turns out that the latter is an object that is interesting on its own, and we can define it in a more elegant way using a new binary composition operator:

\begin{definition}[Well-Nested Shuffle]\label{def:well-nested-shuffle}
For languages $L_1, \ldots, L_n$ respectively over disjoint alphabets $\tS_1, \ldots, \tS_n$, define the \emph{well-nested shuffle} $L_1 \wnshuffle \cdots \wnshuffle L_n$ as the set of well-nested words in $L_1 \shuffle \cdots \shuffle L_n$.
\end{definition}
We observe that well-nested shuffle is associative, so the well-nested shuffle of $n$ languages for any $n > 1$ is well-defined. For well-matched languages $L_1, L_2$, we have $L_1; L_2 \subset L_{1}\wnshuffle L_{2}\subset L_{1}\shuffle L_{2}$.  

\begin{theorem}[Closure]\label{thm:VPLs-closed-under-wnshuffle}
    Visibly pushdown languages are closed under the well-nested shuffle.
\end{theorem}

This is significant, because the same is not true for the standard shuffle operator. If we view programs  $P_1$, \dots, $P_n$ as languages of their behaviours, then the program $P_1 \wnshuffle \cdots \wnshuffle P_n$ is the maximal, in terms of the set of its behaviours, program consisting of arbitrary alignments of behaviours of $P_1$, \dots, $P_n$ that can be executed using a single stack. 

The well-nested shuffle is a construction of interest, independent of the specific problem of constructing lex reductions. We further remark on this in Section \ref{sec:concurrency}. 

\subsection{Visibly Pushdown Lex Reductions as Reductions of The Well-Nested Shuffle}\label{sec:redwn}
When a contextual order $\prec$ defines a reduction of $P_1 \shuffle \cdots \shuffle P_n$, coherency of $\prec$ guarantees that only well-nested words are selected from it. If we can alternatively define a reduction as a restriction of $P_1 \wnshuffle \cdots \wnshuffle P_n$, it will by definition include only well-nested words. There is, however, a major technical challenge in achieving this goal: the construction path of Theorem \ref{thm:vpred-vp-coherent} relies on the validity of Equation \ref{eq:lex}, which in turn relies on the {\em closedness} of the language. $P_1 \wnshuffle \cdots \wnshuffle P_n$ is not closed (w.r.t. $\Ith$), because swapping a pair of adjacent calls or returns may break well-nestedness.

Our alternative construction relies on a new insight: any VPO $\prec$ can be repaired to a \emph{coherent} (Definition \ref{def:coherence}) VPO $\prec'$ such that for any well-matched $\rho, \tau \in \wn(\tS^*)$ such that $\rho \equiv_{\Ith} \tau$, we have $\rho \preceq \tau \Rightarrow \rho \preceq' \tau$. We can construct $\prec'$ as follows: for any $a, b\in \tS$ and $u \in \tS^*$, with $R_k := \bigcup_{j\neq k} \Sr_j$,
\[
    a \prec'_u b \equiv \begin{cases}
        ((a \in R_k \leftrightarrow b \in R_k) \land a \prec_u b) \lor (a \notin R_k \land b \in R_k), &\text{if $u$ has a last pending call from $\tS_k$} \\
        a \prec_u b. &\text{otherwise}
    \end{cases}
\]
\begin{lemma} \label{lem:repair}
    For any VPO $\prec$, there exists a coherent VPO $\prec'$ such that $\red_\prec(P_1 \wnshuffle \cdots \wnshuffle P_n) = \red_{\prec'}(P_1 \shuffle \cdots \shuffle P_n)$ for any well-matched languages $P_1, \ldots, P_n$ over disjoint VP alphabets.
\end{lemma}

This leads to the key theorem of this section that links a family of VPL reductions to the family of visibly pushdown contextual orders by construction:

\begin{theorem}[Visibly Pushdown Lex Reduction]\label{thm:vpred}
    If $\prec$ is a visibly pushdown contextual order and $P_1, \ldots, P_n$ are VPLs, then so is $\red_\prec(P_1 \wnshuffle \cdots \wnshuffle P_n)$.
\end{theorem}

The theorem is a consequence of Theorem \ref{thm:vpred-vp-coherent} (a generic construction via sleep sets) and Lemma \ref{lem:repair} and yields an automaton of $O(q^n q_A n\,2^{|\tS|})$ states, assuming the VPA for each $P_i$ has $O(q)$ states and the VPA representing the contextual order has $O(q_A)$ states. Next, we argue that under the assumption of full commutativity, one can do better than this.

\subsection{Optimized Construction Specific to Hypersafety}\label{sec:oc}

Remarkably, it turns out that one can exploit the assumption of a full commutativity relation, applicable in all hypersafety verification instances, to devise a {\em simple} construction for the VPA, resulting in a more efficient algorithm. Consider two singleton well-matched languages $\{ \rho_1 \}$ and $\{ \rho_2 \}$. Under the assumption of full commutativity between the two alphabets, the lex reduction $\red_\prec(\{ \rho_1 \} \wnshuffle \{ \rho_2 \})$, which consists of exactly one word, can be computed in a greedy manner: If we have correctly interleaved the prefixes $\rho_1', \rho_2'$ of $\rho_1, \rho_2$ as $\sigma$, then the next letter is the smaller one of the letters following $\rho_1', \rho_2'$ that do not break well-nestedness: if $\sigma$ contains a pending call and the last one is from $\tS_i$, then the next letter is not from any $\Sr_j$ such that $j \neq i$.

With modest assumptions on the order, this idea can be applied to the construction of $\red_\prec(P_1 \wnshuffle P_2)$ as well. 
An order $\prec$ is \emph{uniform} w.r.t. VPA $P_1, P_2$ if for any pair of states $q_1, q_2$ from $P_1, P_2$, for any $w \in \tS^*$, $\prec_w$ ranks outgoing letters of $q_1, q_2$ consistently, i.e., if some outgoing letter from $q_1$ is less than some outgoing letter from $q_2$, then all outgoing letters from $q_1$ are less than all outgoing letters from $q_2$.
If the VPAs are directly obtained from code, the only operational restriction is that the {\tt then} and {\tt else} branches of a component have to be ordered the same way against actions of another component. 
The uniformity condition is not practically limiting, as one can always adjust $P_1$ and $P_2$ (while preserving the language) to accommodate the order. For instance, any VPA can be converted in linear time to one such that for every state, the outgoing letters are all internals, all calls, or all returns.
For a uniform VP contextual order, we can construct a compact VPA for $\red_\prec(P_1 \wnshuffle P_2)$ in the same greedy manner:
\begin{theorem}\label{thm:oc}
    Let $\prec$ be a uniform VP order represented by VPA $A$.
Then the lex reduction $\red_\prec(P_1 \wnshuffle P_2)$ is recognized by a VPA with $O(n_1 n_2 n_A n_\Gamma)$ states, where $n_1$, $n_2$, and $n_A$ are the number of states of $P_1$, $P_2$, and $A$ respectively and $n_\Gamma$ is the sum of the size of the stack alphabets of $P_1$ and $P_2$.
\end{theorem}
As outlined in Section \ref{sec:experiments}, one can convert an automaton to a grammar using a standard construction \cite{AlurM2009} or verify this automaton directly. Directly constructing a VPG for the same language is not straightforward, since the contextual information (used by the contextual order) needs to propagate not in a top-down manner but through the linear and hierarchical edges.

\section{Sound Reductions for Concurrency}\label{sec:concurrency}

As an underapproximation of the shuffle, the well-nested shuffle and its lex reductions are classic recursive programs that can always be used for \emph{bug finding} for recursive concurrent programs.

In the hypersafety context, it is sound to consider only behaviours of $P_1 \shuffle \cdots \shuffle P_n$ that are subsets of $P_1 \wnshuffle \cdots \wnshuffle P_n$, since under the assumption that $P_i$'s are fully disjoint, we have
\begin{equation} \label{eq:wnshuffle-sound}
    \{\pre\}\ P_1 \wnshuffle \dots \wnshuffle P_n\ \{\post\} \implies \{\pre\}\ P_1 \shuffle \dots \shuffle P_n\ \{\post\}.
\end{equation}

In the presence of shared memory, there may be dependency between different components. It is then no longer the case that any run of $P_1 \shuffle \cdots \shuffle P_n$ can be {\em soundly} rearranged to be well-nested and condition \ref{eq:wnshuffle-sound} may no longer hold. We outline under what assumptions about $P_i$'s, we can recover this property for shared memory concurrent programs, so the well-nested shuffle can also be soundly \emph{verified} in place of the original recursive concurrent program.

Consider a simple case of tail-recursive programs first. Since a tail-recursion can be transformed to a standard iteration, one intuitively expects this case to be problem-free, and it is. Assume $P_1$ and $P_2$ each make a single call of a tail-recursive function, resulting in runs $\rho_1 \in P_1$ and $\rho_2 \in P_2$. Since a return statement is always a local operation and therefore soundly commutes with any statement from other components, one can commute the return statements to make any interleaving of $\rho_1$ and $\rho_2$ well-nested while preserving semantics. For example, below, the last (pink) return can be commuted to turn the run to a semantically equivalent well-nested one.
\begin{center}
\includegraphics[scale=0.6]{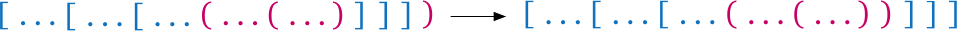}
\end{center}
One can further relax the tail-recursion condition to \emph{tail-independence}.
For any $\tS_i$, let $\tS_i^\indep$ be all letters in $\tS_i$ that soundly commute with all letters in other components.

\begin{definition}[Tail/Head-Independent] \label{def:tail-independent}
    We call a component $P_i$ \emph{tail-independent} (resp. \emph{head-independent}) if every run $\rho$ of $P_i$ can be written as $uv$ ($u, v \in \tS^*_i$) such that
    \begin{itemize}
        \item all calls are in $u$ and all returns are in $v$, except those at a position $j$ such that all letters between $j$ and its matching position are in $\tS_i^\indep$; and
        \item all letters in $v$ (resp. $u$) are in $\tS_i^\indep$.
    \end{itemize}
\end{definition}

This definition is effective, since given alphabets $\tS_i$ and $\tS_i^\indep$, one can construct a VPL $P_i^{\mathsf{max}}$ such that a component $P_i$ is tail-independent iff $P_i \subset P_i^{\mathsf{max}}$. The same holds for head-independence. Hence, for any concurrent recursive program, one can algorithmically check if Defintion \ref{def:tail-independent} applies.

\begin{proposition} \label{prop:tail-indep-vpl}
Given a component $P_i$ as a VPA and the alphabet $\tS_i^\indep$, it is decidable in polynomial time whether $P_i$ is tail-independent (or head-independent).
\end{proposition}

%

We can then conclude conditions under which the well-nested shuffle can be soundly used for verification of a concurrent program:

\begin{theorem}[Soundness of Well-Nested Shuffle Under Concurrency] \label{thm:wnshuffle-sound}
   If every component $P_i$ is tail-independent (or respectively head-independent) 
    then condition \eqref{eq:wnshuffle-sound} is satisfied.
\end{theorem}

The well-nested shuffle can be verified as a classic recursive program using classic verification techniques, without concurrency posing an additional complication. Moreover, as we discuss in Section \ref{sec:redwn}, any sound reduction of it can also be used for the purpose of verification.


\section{Canonical Reductions and Their Construction} \label{sec:canonical}
In Section \ref{sec:reduction}, we introduce lex reductions as an infinite family of recursive product programs. Here, we introduce a simpler (still infinite) subfamily of lex reductions that are defined based on a set of building blocks, called {\em canonical reductions}, which are intuitive for the programmer to understand, compose, and use as annotations in hypersafety verification of recursive programs.

In Section \ref{sec:cr}, we give the formal definitions, and show that this new family is a subset of the lex reductions. As such, the results from the previous section support the construction of these in the style presented in Section \ref{sec:oc}. We observe, however, that for this sub-family, one can follow the inductive definitions and provide a direct VPG construction of a reduction from the family. This construction yields smaller grammars compared to that of Section \ref{sec:oc} as we illustrate in Section \ref{sec:experiments}.

\subsection{Canonical Reductions for Recursive Programs}\label{sec:cr}
We start by giving a formal definition for the nested concatenation, motivated by the \ref{p3} property of $\mathtt{div}$ from Section \ref{sec:intro}, and then one for the {\em parametric lockstep}. We define them on words for simplicity, since one can lift the definition from words to languages (i.e., sets of words) in the standard way. We formally argue that these are instances of lex reductions, and then use an example to discuss the expressive power of the simple subfamily.

\paragraph{\bf Nested Concatenation.} The rules in Fig. \ref{fig:nc} inductively define a ternary relation $w_1 \boldsymbol{\oplus} w_2 \to w$, where $w_1$ and $w_2$ are well-matched words. The \emph{nested concatenation} $w_1 \oplus w_2$ is the word $w$ such that $w_1 \boldsymbol{\oplus} w_2 \to w$, which can be shown to be uniquely defined and well-matched. We generalize this from two to $n$ words by letting $w_1 \oplus w_2 \cdots \oplus w_n = w_1 \oplus (w_2 \oplus \cdots \oplus w_n)$.

\begin{figure}[h]
\begin{center}
\begin{tabular}{|ccc|}\hline
	\infer[\mbox{NC-}\epsilon]{\epsilon \boldsymbol{\oplus} w_2 \to w_2}{} \ \ \ \ 
&
	\infer[\mbox{NC-int}]{{\blue a}w_1 \boldsymbol{\oplus} w_2 \to {\blue a}w'}{
 w_1 \boldsymbol{\oplus} w_2 \to w'}\ \ \ \ 
&	\infer[\mbox{NC-call}]{{\blue c}w_1{\blue r}v \boldsymbol{\oplus} w_2 \to {\blue c}w'{\blue r}v}{
 w_1 \boldsymbol{\oplus} w_2 \to w'}\\\hline
\end{tabular}
\end{center}\vspace{-5pt}
\Description{Definition of nested concatenation}
\caption{Definition of nested concatenation: $\blue{c}$, $\blue{r}$, and $\blue{a}$ respectively stand for a call, a return, or an internal statement. In ${\blue c}w{\blue r}v$, the $\blue{r}$ is meant to be the return that matches $\blue{c}$. \label{fig:nc}}\vspace{-5pt}
\end{figure}


\paragraph{\bf Parametric Lockstep.} Next, we define a lockstep product that is parametric on a vector $\vec{s}$ of positive integers that provides the nominal speed of the movement of each component in the lockstep synchronization. We define a simple {\em modulo} decrement operation for such vectors to simplify our definition. For vectors of natural numbers $\vec{s}$ and $\vec{t}$ such that $|\vec{s}| = |\vec{t}| = n > 0$, $\vec{0} \le \vec{t} \le \vec{s}$ (componentwise), and $\vec{t} \neq \vec{s}$, define
\[
    \dec_{\vec{s}}(\vec{t}) := \begin{cases}
        (s_1 - 1, s_2, \ldots, s_n), &\text{if } \vec{t} = \vec{0} \\
        (0, \ldots, 0, t_i - 1, t_{i+1}, \ldots, t_n). &\text{if } i = \min\{j:t_j > 0\}
    \end{cases}
\]

\def\LSS#1{\mathsf{LS}_{\vec{#1}}}
\def\LS#1#2{\mathsf{LS}_{\vec{#1},\vec{#2}}}

$\LSS{\vec{s}}(w_1, \dots, w_n)$ denotes the lockstep composition of the $n$ well-matched words at the {\em speed} of $\vec{s}$. To inductively define it, we need an additional helper vector $\vec{t}$.  
The rules in Fig. \ref{fig:ls} define a relation $\LS{s}{t}(w_1, \dots, w_n) \to w'$, where $\vec{s}$ and $\vec{t}$ are vectors satisfying the constraints above. The \emph{$\vec{s}$-lockstep} of $w_1, \ldots, w_n$ is the unique word $w'$ such that $\LS{s}{0}(w_1, \dots, w_n) \to w'$. 

Note that the definition is inductive over the structure that exposes the \emph{leftmost} letter of a word. This is intentional, for consistency with the definition of CLO (Definition \ref{def:clo}), which looks at the \emph{left context}. However, this is not the only valid choice, and we discuss the alternative in Section~\ref{sec:customizations}.

\begin{figure}[t]
\begin{mathpar}
\infer[\mbox{LS-single}]{\LS{s}{t}(w_1) \to w_1}{}
    \and
\infer[\mbox{LS-}\epsilon]{\LS{s}{t}(w_1, \dots, w_{j-1}, \epsilon, w_{j+1}, \dots w_n) \to w'}
{\LS{s'}{t'}(w_1, \dots, w_{j-1}, w_{j+1}, \dots w_n) \to w' & 
\vec{s'},\vec{t'} \leftarrow \vec{s}_{-j},\vec{t}_{-j}
}

\and

\infer[\mbox{LS-int}]{\LS{s}{t}(w_1, \dots, w_{m-1}, \blue{a}w_{m}, w_{m+1}, \dots, w_n) \to w'}
{  
\LS{s}{t}(w_1, \dots, w_n) \to w'
&
\text{none of $w_i$ where $i<m$ starts with an internal}
}

\and

\infer[\mbox{LS-call-1st}]{\LS{s}{t}(\blue{c_1}w_1\blue{r_1}v_1, \dots, \blue{c_n}w_n\blue{r_n}v_n) \to \blue{c_1}u'\blue{r_1}v'}
{\LS{s}{{\dec_{\vec{s}}(\vec{t})}}(w_1, \blue{c_{2}}w_2\blue{r_{2}}, \dots, \blue{c_{n}}w_{n}\blue{r_{n}}) \to u' 
& 
\vec{t} = \vec{0} 
&
\LS{s}{t}(v_1, \dots, v_n) \to v'
}
\and

\infer[\mbox{LS-call}]{\LS{s}{t}(\blue{c_1}w_1\blue{r_1}v_1, \dots, \blue{c_n}w_n\blue{r_n}v_n) \to \blue{c_m}u'\blue{c_m}v_m}
{\LS{s}{{\dec_{\vec{s}}(\vec{t})}}(\dots, \blue{c_{m-1}}w_{m-1}\blue{r_{m-1}}, w_{m}, \blue{c_{m+1}}w_{m+1}\blue{r_{m+1}}, \dots) \to u' 
& 
\vec{t} \neq \vec{0} 
&
m = \min\{i : t_i > 0\}
}

\end{mathpar}\vspace{-10pt}
\Description{Definition of $\vec{s}$-lockstep}
\caption{Inductive defintion of parametric lockstep. $\vec{s}_{-j}$ denotes the vector obtained from $\vec{s}$ by dropping the $j$-th component, and similar for $\vec{t}_{-j}$. \label{fig:ls} }\vspace{-10pt}
\end{figure}

\paragraph{\bf Subfamily of Lex Reductions.}
Crucially, we can show that the nested concatenation and $\vec{s}$-lockstep of well-matched VPLs over disjoint alphabets can be produced through a VP contextual order (VPO) as a lex reduction of the well-nested shuffle. 
\begin{proposition}\label{prop:cr-vplex}
    If $P_1, \dots, P_n$ are well-matched VPLs over disjoint alphabets, then their concatenation, nested concatenation, and $\vec{s}$-lockstep are lex reductions of $P_1 \wnshuffle \cdots \wnshuffle P_n$ induced by VPOs.
\end{proposition}
We point out that the VPOs used by the construction are uniform (see Section \ref{sec:oc}) w.r.t. VPA such that for every state, the outgoing letters are all internals, all calls, or all returns. Combining Proposition \ref{prop:cr-vplex} and Theorem \ref{thm:vpred}, we have the following theorem:
\begin{theorem}\label{thm:cr-vp}
    If $P_1, \dots, P_n$ are well-matched VPLs over disjoint alphabets, then any reduction built from a combination of concatenation, nested concatenation, and $\vec{s}$-lockstep and one occurrence of each $P_i$ is a VP reduction of $P_1 \wnshuffle \cdots \wnshuffle P_n$. 
\end{theorem}

\paragraph{\bf Expressive Power.}
Superficially, nested concatenation and parametric lockstep may seem simple and trivial. As we illustrate empirically in Section \ref{sec:experiments}, their compositions yields a powerful family that can solve nearly all our benchmarks.
However, to discuss this in depth, let us look at an example borrowed from \cite{ChurchillPSA2019}, which puts forward a solution for constructing {\em semantic} product programs. Our reductions are syntactic and as we remark in Section \ref{sec:discussion}, this implies a strict limitation compared to semantic ones. However, as we demonstrate using the example below, the combination of canonical reductions with slight \emph{customizations}, which operate on \emph{unbounded} and \emph{rich} syntax, can be surprisingly powerful.

\begin{wrapfigure}[12]{r}{0.5\textwidth}\vspace{-15pt}
\includegraphics[scale=0.45]{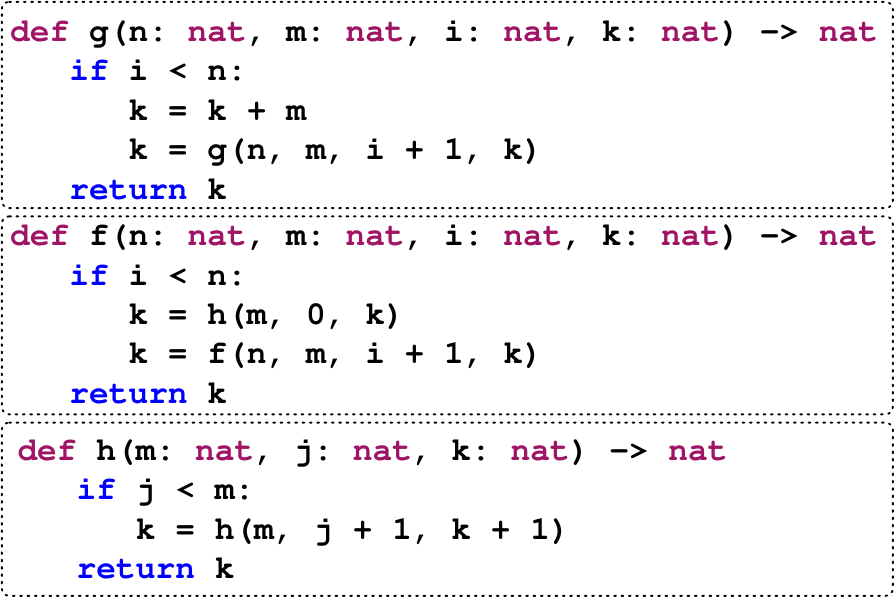}\vspace{-8pt}
\caption{Challenging Example From \cite{ChurchillPSA2019} \label{fig:12}}
\end{wrapfigure}
Consider the functions $f$ and $g$ illustrated on the right and the goal of proving them equivalent. They both increment $k$ by $m$ for $n$ times; the only difference is that in $f$, the number $m$ is computed by an auxiliary function $h$ as the sum of $m$ $1$'s.
They appear in Figure 12 from \cite{ChurchillPSA2019}, where an iterative version of this is cited as a challenge for verification using product programs, even {\em semantically} constructed ones. To solve this in \cite{ChurchillPSA2019}, the parameter $m$ is bounded to a small constant.

In constrast, our methodology easily yields a solution. In each interleaving run, we intuitively want to match every call of $f$ against one call of $g$. This is \emph{almost} a $(1, 1)$-lockstep, except that the formal definition in Fig. \ref{fig:ls} dictates that \emph{every} call and return participates in the scheduling, which undesirably forces the calls of $h$ to also participate in the lockstep matching of the calls; as a result, the second recursive call of $g$ are nested inside the first call of $h$. However, as we demonstrate in Section \ref{sec:customizations}, this limitation in the formal definitions of the canonical reductions is superficial, and we can easily \emph{customize} them to schedule only a selected subset of calls and returns (in this case, those of functions $f$ and $g$) and treat the rest the same way as internal actions. Our tool {\tool} can verify this instance when such a customized $(1,1)$-lockstep reduction is given as input.

\subsection{Customizations of Canonical Reductions} \label{sec:customizations}
In this section, we introduce two simple \emph{customizations} of canonical reductions.

\paragraph{\bf Scheduling Selected Calls/Returns Only.}

By default, every call and return participates in the scheduling of the canonical reductions, and our first customization is to allow only a selected subset to participate, and the rest stack operations $\Sigma' \subset \Sc \cup \Sr$ are treated the same way as internal actions.

To achieve this, one cannot simply change the partition of the visibly pushdown alphabet, which breaks the mechanics of visibly pushdown languages; nor can we apply the NC-int rule or LS-int rule for runs that start with a call that does not participate in the scheduling, which breaks well-nestedness of the reduction. Instead, we modify the \emph{contextual orders} $\prec$ for the canonical reductions (Proposition \ref{prop:cr-vplex}): Based on $\prec$, we can define another contextual order $\prec'$ that replaces letters in $\Sigma'$ with internal letters, queries $\prec$, and returns the result. Then the reduction induced by $\prec'$ is exactly the customization that we want. Formally, let $f : \tS^* \to \tS^*$ be a letter-to-letter string homomorphism such that
\begin{equation}
    f(a) \in \Sl_i \text{ if } a \in \Sigma' \cap \tS_i \text{ for any } i, \text{ and } f(a) = a \text{ otherwise.} \label{eq:homo}
\end{equation}
Then $\prec'$ is given by the contextual order $f(\prec)$ defined by
$\forall a, b \in \tS.\; a\; f(\prec)_w\; b \Leftrightarrow f(a) \prec_{f(w)} f(b)$.

In general, $f(\prec)_w$ is not a linear order even when $\prec$ is a map $\tS^* \to \Lin(\tS)$, as we have been assuming; but in our use case, where $f$ replaces excluded calls and returns with internals from the same component, such a definition is sufficient to induce a minimal reduction. Moreover, when $\prec$ is a VPO (as in Proposition \ref{prop:cr-vplex}), with a reasonable assumption on $\Sigma'$, there is a contextual order \emph{equivalent} to $f(\prec)$ that is a VPO:
\begin{proposition} \label{prop:homo-vpo}
    Let $\Sigma'$ be a subset of $\Sc \cup \Sr$ such that for all $i \in \{1,\ldots,n\}$ and $c, r \in \tS_i$ that match in some run of $P_i$, we have $c \in \Sigma' \Leftrightarrow r \in \Sigma'$. Let $f$ be a letter-to-letter string homomorphism that satisfies \eqref{eq:homo}. If $\prec$ is a VPO, then there is a VPO $\prec'$ such that $\red_{\prec'}(P_1 \wnshuffle \cdots \wnshuffle P_n) = \red_{f(\prec)}(P_1 \wnshuffle \cdots \wnshuffle P_n)$.
\end{proposition}
The proof in \refapp{app:proofs} gives a construction for the VPO $\prec'$ given $\prec$ and $f$.
\paragraph{\bf Right Alignment.}
The canonical reductions are defined over the structure that exposes the \emph{leftmost} letter of a word, for consistency with the use of \emph{left context} (Definition \ref{def:clo}), which eventually leads to Proposition \ref{prop:cr-vplex}. However, this does not have to be the case: we can customize the \emph{alignment} of canonical reductions. To have \emph{right-aligned} canonical reductions, we consider well-matched words of the form $\epsilon$, $w\blue{a}$, and $v\blue{c}w\blue{r}$ and change the rules accordingly. Then Proposition \ref{prop:cr-vplex} holds for the modified version of contextual lex reductions that look at the \emph{right context}.

As an example, consider a benchmark from \cite{MordvinovF2019} about proving the monotonicity of the \emph{Ackermann function} $A(m, n)$ over the second argument:
\[
    A(m, n) := \begin{cases}
        n + 1, &\text{if } m = 0 \\
        A(m - 1, 1), &\text{if } m > 0 \land n = 0 \\
        A(m - 1, A(m, n - 1)). &\text{if } m > 0 \land n > 0
    \end{cases}
\]
Intuitively, a $(1, 1)$-lockstep reduction seems like the appropriate choice. However, let us consider a run $\rho_1$ that starts with calling $A(m, 0)$ and a run $\rho_2$ that starts with calling $A(m, n)$ ($m, n > 0$). The run $\rho_1$ consists of one direct recursive call $A(m - 1, 1)$, whereas the run $\rho_2$ consists of two direct recursive calls $a := A(m, n - 1)$ and $A(m - 1, a)$, in that order. We would want to match $A(m - 1, 1)$ in $\rho_1$ against the \emph{second} call $A(m - 1, a)$ in $\rho_2$, since it is the second call that produces the final result and has the same argument $m - 1$. A concise way to describe such a reduction is the \emph{right-aligned} $(1, 1)$-lockstep.

\subsection{Direct Construction of Canonical Reductions} \label{sec:eccr}
Given well-matched VPGs $G_1, \ldots, G_n$ describing the language of the syntactic runs of each component, the inductive definitions of the canonical reductions (Figures \ref{fig:nc} and \ref{fig:ls}) inspire a way to directly construct them as VPGs.
A syntactic assumption on the form of input VPGs facilitate this construction: for any productions $X \to \alpha$ and $X \to \beta$, we assume that the right-hand sides $\alpha, \beta$ both start with a call, both start with an internal, or both are empty. Any VPG can be converted in linear time to an equivalent one in this form. 


\paragraph{\bf Nested Concatenation:} One can view the construction as a type of {\em product} construction.  
The nonterminal symbols of the grammar $G$ for the language $\Lang(G_1) \oplus \Lang(G_2)$ are of the form $[W_1,W_2]$, where $W_1,W_2$ are nonterminal symbols of $G_1,G_2$ respectively; $[W_1,W_2]$ is a start symbol if $W_1,W_2$ are start symbols of $G_1,G_2$ respectively. The productions of $G$ are defined by the rules in Figure \ref{fig:nc-vpg}, where $(X \to \alpha) \in G$ denotes $(X \to \alpha)$ is a production in $G$, a slight abuse of notation. Observe the syntactic similarity between the rules of Figure \ref{fig:nc} and those in Figure \ref{fig:nc-vpg}. The extra syntactic condition guarantees this simplicity of inference of a grammar from the inductive definition. 
\begin{figure}[h]
\footnotesize
\begin{mathpar}
    \inferrule
    {(Y_1 \to \epsilon) \in G_1\\\\
    (W_2 \to \alpha) \in G_2}
    {([Y_1,W_2] \to \alpha) \in G}\text{GN-}\epsilon
    \and
    \inferrule
    {(Y_1 \to \blue{a}W_1) \in G_1\\\\
    (W_2 \to \alpha) \in G_2}
    {([Y_1,W_2] \to \blue{a}[W_1,W_2]) \in G}\text{GN-int}
    \and
    \inferrule
    {(Y_1 \to \blue{c}W_1\blue{r}V) \in G_1\\\\
    (W_2 \to \alpha) \in G_2}
    {([Y_1,W_2] \to \blue{c}[W_1,W_2]\blue{r}V) \in G}\text{GN-call}
\end{mathpar}\vspace{-10pt}
\Description{VPG construction rules of the nested concatenation}
\caption{VPG productions for nested concatenation \label{fig:nc-vpg}}\vspace{-10pt}
\end{figure}

\paragraph{\bf $\vec{s}$-Lockstep:} 
The construction of $\vec{m}$-lockstep follows the same pattern. The nonterminal symbols of the grammar $G$ of the $\vec{m}$-lockstep of $\Lang(G_1), \ldots, \Lang(G_n)$ have the form $[\vec{s},\vec{t},(W_1, \dots, W_k)]$, where $\vec{s},\vec{t}$ satisfy the constraints in Section \ref{sec:cr}, $W_i$ is a nonterminal in $G_i$ for each $i$, and the length and maximum element of $\vec{s}$ are bounded by those of $\vec{m}$. The rules for generating the grammar construction rules mimic those of the inductive definition in Figure \ref{fig:ls}.

There is, however, a small complication with an easy fix. In the rules for nested concatenation (in Figure \ref{fig:nc}), for every pattern in the conclusion, the parts that appear in the premises have a nonterminal symbol dedicated to it. For instance, $\blue{a}w_1$ is a pattern in the inductive definition, and $w_1$ corresponds to nonterminal symbol $W_1$ in the grammar construction rules. This makes things straightforward. Now consider the rule ``LS-call'' (of figure \ref{fig:ls}).

There is a pattern $\blue{c_i}w_i\blue{r_i}v_i$ in the conclusion, yet the subword $\blue{c_i}w_i\blue{r_i}$ is isolated in the premise, for which there does not exist a nonterminal. We can fix this if, for every rule of the form $Y \to \blue{c}W\blue{r}V$ in $G_i$, we add a rule $Y' \to \blue{c}W\blue{r}E$ (where $E \to \epsilon$), so subwords generated by $\blue{c}W\blue{r}$ can be referred to as $Y'$. The full set of production rules are listed in Figure \ref{fig:ls-vpg} in \refapp{app:cr-details}.

\paragraph{\bf Algorithmic Construction:} In an algorithmic view of this construction, we can maintain a queue of reachable symbols, so only the productions that are truly needed are added to the grammar, making the construction efficient. Otherwise, the grammar would end up having a lot of unreachable non-terminals. These VPG constructions exploit properties of the canonical reductions: the contextual information is propagated in a top-down manner. For an arbitrary VP contextual order, the contextual information propagates through the linear and hierarchical edges of the nested words, so the reduction does not necessarily admit a direct VPG construction in the style demonstrated in this section.

\section{The Semantics of The Recursive Product Program}\label{sec:semantics}
Recall that our full notation for a recursive program is a pair $(P,\mathcal{S})$  with $P$ representing its syntax as a visibly pushdown language and $\mathcal{S}$ representing its semantics. 
Let $(P_1, \mathcal{S}_1), \ldots, (P_n, \mathcal{S}_n)$ be recursive programs over disjoint VP alphabets $\tS_1, \ldots, \tS_n$ respectively. Assume that we have constructed a product program $\nP$ that is sound, well-nested, and visibly pushdown. To complete its description, we need to discuss its semantics. Here, we discuss how to compose $\mathcal{S}_1, \ldots, \mathcal{S}_n$ into a single $\mathcal{S}$. 

We assume the variable names from the individual programs are disjoint. Suppose $\mathcal{S}_k$ is induced by $\mathcal{F}_k$ (see Section \ref{sec:rp}) and $V = V_1 \uplus \cdots \uplus V_n$ where $V_k$ is the set of all valuations for $P_k$. The trivial solution would be to use $n$ stacks, which corresponds to a semantics $\mathcal{M}: \tS \to \powerset(\Pi_{j=1}^n \stack(V_j) \times \Pi_{j=1}^n \stack(V_j))$. Each component uses the corresponding stack from the tuple, and its semantic is defined in the standard way.
For an interleaving run $\rho \in P_1 \shuffle \cdots \shuffle P_n$, the Hoare triple $\{\pre\} \rho \{\post\}$ is {\em valid} w.r.t. $\mathcal{M}$ if 
\[\forall ((S_1.\nu_1, \ldots, S_n.\nu_n), (S'_1.\nu'_1, \ldots, S'_n.\nu'_n)) \in \sem{M}{\rho}:\ \bigcup_{j=1}^n \nu_j \models \pre \implies \bigcup_{j=1}^n \nu'_j \models \post.\] 
$\mathcal{M}$ is the semantics where all commutativity reasoning happens: the commutativity relation $\Ith$ (see Section \ref{sec:contextual-lex-reductions}) is sound w.r.t. $\mathcal{M}$.


By construction, any product program $\nP$ is well-nested, so it conceptually runs on a single stack. However, if we just naively push and pop on the same stack, we run into a problem with scoping of variables from different components. For instance, suppose $P_1$ pushes and writes to variable $x_1$; if $P_2$ pushes next, then $P_1$ loses track of $x_1$, which is no longer on the top frame.

There is a simple solution to this problem. We interpret the letters slightly differently: let each frame in the single stack store the variables from \emph{all} $n$ recursive programs (recall that the variable names are assumed to be disjoint); when $P_k$ pushes or pops the stack, we also copy over the irrelevant variables, i.e., those from $P_j$ where $j \neq k$. Formally, we define a semantics $\mathcal{S}: \tS \to \powerset(\stack(V) \times \stack(V))$ induced by $\mathcal{F} \subset (\Sc \times V \times V) \cup (\Sr \times V \times V \times V) \cup (\Sl \times V \times V)$:
\begin{itemize}
    \item If $c \in \Sc_k$, then $\sem{F}{c} := \{ (\nu, \nu) : \nu \in V_j,\, j\neq k \} \cup \sem{F_\mathit{k}}{c}$.
    \item If $r \in \Sr_k$, then $\sem{F}{r} := \{ (\nu, \nu_{<}, \nu) : \nu, \nu_{<} \in V_j,\, j\neq k \} \cup \sem{F_\mathit{k}}{r}$.
    \item If $a \in \Sl_k$, then $\sem{F}{a} := \{ (\nu, \nu) : \nu \in V_j,\, j\neq k \} \cup \sem{F_\mathit{k}}{a}$.
\end{itemize}
$\mathcal{S}$ can be easily adjusted to account for shared memory: the shared global variables are also copied over when when the stack is pushed and popped.

Recall that $\mathcal{M}$ is the $n$-stack semantics. Well-nestedness guarantees that the shift from $\mathcal{M}$ to $\mathcal{S}$ preserves the validity of Hoare triples, so any run in the well-nested shuffle has no distinction from a run of a standard recursive program.
\begin{theorem} \label{thm:sem}
    For any syntactic run $\rho \in P_1 \wnshuffle \cdots \wnshuffle P_n$ and assertions $A, B$, the validity of $\{A\} \rho \{B\}$ is the same with respect to the semantics $\mathcal{M}$ and $\mathcal{S}$. 
\end{theorem}


\section{Experimental Results} \label{sec:experiments}
We have implemented our reduction-based scheme as a tool, called {\tool}, that verifies hypersafety properties of sequential recursive programs. 
\subsection{Implementation} \label{sec:implementation}

{\tool} verifies an input program written in a simple procedural programming language against a given pair of pre/postconditions. The programming language supports the integer, boolean, and array data types. The user specifies the property and one of the built-in (parametric) canonical reductions, or a combination of canonical reductions. For instance, to verify the \ref{p3} property of $\mathtt{div}$, one would use the command\footnote{\ref{p3} is a shorthand for the pre/postcondition that expresses the property.}:
\begin{equation}
 \text{\footnotesize
 \tool\ {\tt div -property} \ref{p3} {\tt -reduction "(1,1)-lockstep(P3, nested\_concatenation(P1,P2))"}} \tag{C}
 \label{eq:cmd}
\end{equation}
{\tool} also accepts an arbitrary uniform visibly pushdown order, but none was required to solve any of the benchmarks used for the evaluation.

{\tool} implements the following algorithms:
\begin{description}
  \item[\sc Aut:] the optimized construction of Section \ref{sec:oc} as a nested word automaton (NWA) \cite{AlurM2009}, a model that is equivalent to a VPA and supported by existing libraries \cite{HeizmannCDE+2013};
  \item[\sc VPG:] the same optimized construction but additionally converting the NWA into a VPG;
  \item[\sc Direct:] the direct construction of Section \ref{sec:eccr}.
\end{description}

In each case, the resulting product program (in NWA or VPG format) and the corresponding hypersafety property are encoded by {\tool} as system of constrained Horn clauses (CHCs) in the SMT-LIB format \cite{BarST-SMT-10}, which is satisfiable if and only if the pre/postcondition holds. We did not find a classic recursive verification tool that does well and is not already using a similar CHC encoding in the backend. In particular, \cite{HeizmannHP2009} cannot handle these instances well. We use Z3/Spacer~\cite{Z3,KomuravelliGC2016}, Eldarica \cite{HojjatR2018}, and Golem \cite{BlichaBS2023} as the backend CHC solvers.

\paragraph{\bf CHC Encodings.}

We interpret the product program under the semantics $\mathcal{S}$ in Section \ref{sec:semantics}. An NWA can be directly encoded as CHCs in linear time (see \refapp{app:enc}): one predicate for each state, and one constraint for each transition that captures the \emph{inductivity} \cite{HeizmannHP2010} of the predicates (assertions). In \cite{HeizmannHP2010}, the assertions annotate a nested trace, whereas in our setting, the assertions annotate the NWA, and the naive encoding would be unsound with respect to \textsf{unsat} since it does not distinguish context-sensitive and insensitive paths. Our solution is to introduce a ghost variable that represents the ``return address'', so only assertions along context-sensitive paths are required to be inductive.


In the conversion of the NWA to a VPG, we use the standard quadratic time \cite{AlurM2009} algorithm and then eliminate the useless symbols and productions \cite{HopcroftMU2007}. A VPG is encoded as CHCs in linear time (see \refapp{app:enc}): one predicate for each nonterminal that \emph{overapproximates} the semantics of the language generating by this nonterminal. The semantics $\mathcal{S}$ involves a stack, but for a well-matched language, if we are only concerned with partial correctness, then we may use a \emph{stack-free} semantics defined inductively over the nesting structure, relating the variables before and after the execution of a trace. VPGs have this nesting structure built-in, and each production is translated into one constraint. This encoding is similar to ``transition summaries'' from \cite{BjornerGMR2015}.

\subsection{Benchmarks}
We built a benchmark set by including the safe relational benchmarks from \cite{MordvinovF2019} (28 in total, which are in turn adapted from \cite{DeAngelisFPP2016,MordvinovF2017}), taking the iterative hypersafety (non-concurrent) benchmarks from \cite{FarzanV2019} (17 in total) and rewriting them as recursive code, and adding 58 new ones. In some cases, the new additional instances are in the form of additional hypersafety properties of benchmarks from \cite{MordvinovF2019,FarzanV2019}. Our newly added benchmarks consist of two groups, one for arithmetic functions (33), and one for functions on arrays (25). In total, we have 62 arithmetic benchmarks, and 41 array benchmarks. (See \refapp{app:benchmarks} for details.)

The arithmetic benchmarks are hypersafety properties of primitives, in the spirit of the $\mathtt{div}$ example from Section \ref{sec:intro}. They include programs with linear or non-linear recursion, and then several hypersafety properties per program. Recall that the \ref{p1} property of $\mathtt{div}$ is relatively straightforward, but the \ref{p3} property of it is rather challenging to solve. So, the hardness of an instance can depend on the benchmark or the property or both. 
The array benchmarks are hypersafety properties of, or relations between, aggregate functions (min, max, and sum) and comparators (pointwise comparison and lexicographic order) on arrays; we have also included benchmarks that manipulate trees that are encoded as heaps in arrays; this is due to restrictions in the CHC solvers' support for theories and has nothing to do with our methodology. 

Despite the fact that hyperproperties were originally motivated by their applications in security, in terms of diversity and quantity, there are many more instances of {\em primitives} (like division or lexicographical orders on vectors) implementing known operations where $k$-safety properties are essential to encode and check the basic properties of these primitives. The arities of the hypersafety instances for these can vary. All but four benchmarks are 2- or 3-safety properties, and the rest has $k=7$ at maximum. These properties are specified in the style of \eqref{eq:cmd}.

\subsection{Experiments}
Our experiments were designed to answer the following research questions:
\begin{itemize} 
\item[RQ1:] Is our reduction approach essential and effective in proving hypersafety properties of recursive programs?
\item[RQ2:] Are parameterized lockstep and nested concatenation, as building blocks,  expressive enough to support the construction of effective product programs for a broad range of examples?
\item[RQ3:] Does the representation of reduction (automata v.s. two grammar styles) have an impact on the effectiveness of the process?
\end{itemize}

With respect to concurrency, this paper mainly proposes the theoretical insight that our reduction methodology acts as a complete sequentialization scheme for certain classes of recursive programs. We do not view {\tool} as a tool for verification or bug finding of recursive concurrent programs.

\paragraph{\bf Results.}
As a {\bf baseline} comparison, we use three standard scenarios, with the \emph{VPG encoding} in \refapp{app:enc}: (1) sequential composition as the product program, (2) a direct encoding of the problem without reductions and without making $k$ disjoint copies for the $k$-safety property, and (3) the variation of (2) in which $k$ disjoint copies are made. We refer to {\sc Baseline} as the best of these three algorithms, including the best performing CHC solver. We tried {\sc RelRecMc} (the tool from \cite{MordvinovF2019}) on all baseline encodings, but it times out or throws a \texttt{Segmentation fault} on every instance. We also tested the $\mathtt{div}$ examples from Section \ref{sec:intro} with its own style of encoding, and \ref{p2} and \ref{p3} cannot be solved.

On average, the production of the CHC encoding takes very little time (around 1 second) for all instances. As expected, the direct construction method is the fastest, followed by the generation of the automaton and the VPG based on the automaton. Our reported runtimes, therefore, are only for the time to verify the product program, and excludes this negligible construction time. We use a timeout of 600 seconds for the combination of both steps.

\begin{figure}[t]
\begin{center}
\includegraphics[scale=0.85]{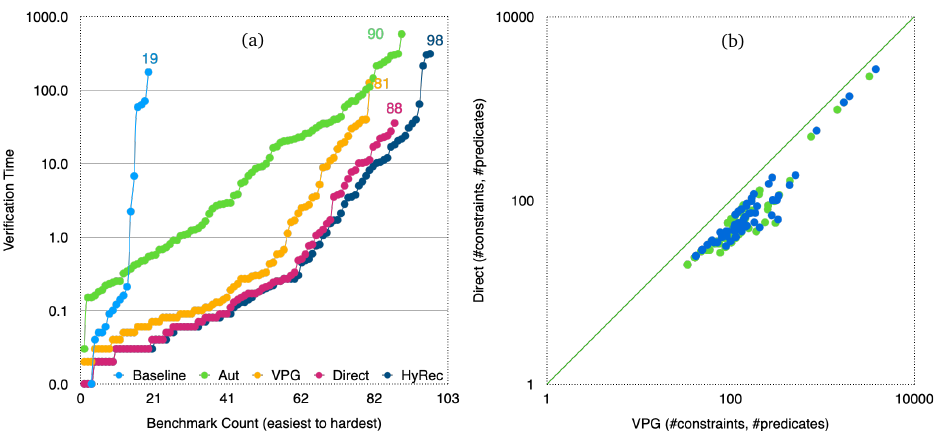}
\caption{Experimental Results: logscale plots. (a): quantile plot of performance comparison {\tool}'s algorithms against the baseline. (b): scatter plot of comparing CHC encoding sizes for the two grammar constructions.}
\label{fig:plot}
\end{center}\vspace{-15pt}
\end{figure}

Figure \ref{fig:plot}(a) illustrates the {\em log-scale} quantile plot that compares our three algorithms ({\sc Aut}, {\sc VPG}, and {\sc Direct}) against the {\sc Baseline}. {\tool} is the best of the three algorithms. The arithmetic benchmarks solved by {\sc Baseline} are exactly those with a function summary in Linear Integer Arithmetic. (RQ1) {\sc Direct}, {\sc VPG}, and {\sc Aut} contribute 73, 17, and 13 instances to {\tool} respectively (accounting for the ties). The number for each curve lists the total number of benchmarks solved by that algorithm. It is noteworthy that CHC encodings based on grammars are on average much more efficiently solvable than the one based on automata; yet, {\sc Aut} is better at solving harder instances: out of the 13 winners contributed by {\sc Aut}, 9 of them are only solvable by {\sc Aut}, whereas all 17 benchmarks where {\sc VPG} wins can be solved by others albeit more slowly. (RQ3)

Our direct construction does offer a significant improvement. This is while the resulting product programs are identical when it comes to the set of their syntactic runs. The difference may stem from the compactness of the resulting grammar. Figure \ref{fig:plot}(b) compares the sizes of the resulting CHC encodings:  blue dots compare the number constraints and green dots that of predicates. (RQ3)

{\tool} fails to solve 5 of the 103 benchmarks. One instance is a shortcoming of the CHC solvers and is not related to the product program aspect. It is the simplest of a group of 4 benchmarks from \cite{FarzanV2019} that increasingly unroll loops. {\tool} successfully proves unrolling the loop by 3, 4, and 5 times, but fails at the 2 unrolling. The rest of the failures are worthy of a deeper discussion, and are explored in the next section. (RQ2)

Separately, we also evaluate the scalability of our approach for higher arity hypersafety instances. With increasing $k$ until the first timeout or error, we run the tool with the $\mathtt{div}$ program (Fig.~\ref{fig:iexample}), $(1,\ldots,1)$-lockstep reduction, and the variation of \eqref{p1} as a $k$-property: \[ n_1 \le \cdots \le n_k \Rightarrow \mathtt{div}(n_k, d) \le \cdots \le \mathtt{div}(n_k, d). \]
As the arity of the hypersafety instance goes higher, the time for generating the CHC encoding is increasing. The three algorithms {\sc Direct}, {\sc Aut}, and {\sc VPG} scale to $k = 12$, $7$, and $6$ respectively; for the largest instance ({\sc Direct}, $k = 12$, and $10^5$ predicates in the CHC encoding), CHC generation takes up $33\%$ of the total time. In practice, however, these arities for hypersafety are very rare. As expected, the size of the CHC encodings (measured by the number of predicates or constraints) grows multiplicatively for all algorithms.

\subsection{Discussion and Limitations}\label{sec:discussion} 
Here, we summarize some key observations from the experiments and discuss some unique features and limitations of our proposed solution. Most importantly, we discuss the reasons behind why 4 benchmarks from the set cannot be solved by {\tool}.

\paragraph{\bfseries Representation matters.} It is somewhat surprising that although the same exact set of alignments are represented by the two different grammar constructions and the automata construction, the results significantly differ in some cases. The difference between the grammar constructions can be attributed to the size, as the scattered plot in Figure \ref{fig:plot}(b) highlights. This advantage, of having 3 chances at solving an instance, is a strict consequence of our {\em operational} framework. 

\paragraph{\bfseries Alignments do not always work.} It is undeniable that reductions are essential to the success of these hypersafety verification instances. It is also important to note that not all hypersafety verification instances can be solved using these reductions. For instance, we use a primitive $\mathtt{mult}(a,b)$ in our benchmark set that multiplies $a$ by $b$ through repetitive addition. Consider the following simple theorem about the commutativity of multiplication: $\mathtt{mult}(n,m) = \mathtt{mult}(m,n)$. This is a 2-safety property, where none of the lex reductions in this paper can simplify the proof down to a decidable theory (linear integer arithmetic). The same is true for the iterative version of this code as well. Two benchmarks from the set cannot be solved for this reason.  

\paragraph{\bfseries Explicit Product Programs v.s. Integrated Verification Technique}
There is a benchmark from \cite{MordvinovF2019} related to Fibonacci that {\tool} cannot solve, which stems from the fact that the conceptual product program for this must reorder the two recursive calls to succeed. In code (and standard semantics), this reordering is forbidden, but at the level of CHC encoding, once one is looking at a symbolic encoding, it can be done semantically and soundly. One can weaken program semantics to permit such sound reorderings. However, when the manipulations are done at the level of the symbolic encoding, it comes for free. As discussed in Section \ref{sec:intro}, having an explicit product program can have other advantages: one can get testing, symbolic execution, program analysis, and other off-the-shelf applications for free. 

\paragraph{\bfseries Syntactic v.s. Semantic alignments.}
Our reduction are syntactic, and {\tool} fails at solving the {\em sum of squares} benchmark \cite{ShemerGSV2019} (modeled by an array in our benchmark set), which truly requires a semantic alignment in the sense that one needs access to the program data to determine the right synchronization; the tool from \cite{ItzhakySV2024}, for iterative programs, can handle the iterative version of this benchmark automatically. We believe semantic alignments can be incorporated in our framework, but this requires a nontrivial shift in representation, for instance to {\em symbolic automata} \cite{DAntoniV2017,DAntoniA2014}, and this is left for future work. It is noteworthy, however, that in some instances in the literature, {\em semantic} alignments are used to account for limitations in {\em syntactic} alignments that are superficial. For instance, our nested concatenation is an instance of a syntactic benchmark that may seem superficially like a semantic benchmark: one needs to wait for one recursive call to be finished before starting the second one (which is meant to be nested inside it). However, we handle this in a purely syntactic way, and lex reductions in general include reductions in which the two copies can go out of synch for arbitrarily long durations, without the need to rely on program data to make this happen.

\paragraph{\bfseries Programmer's intuition goes a long way.} We found it very easy and straightforward to guess and specify the right choice of reduction for each benchmark. Considering that 45\% of our benchmarks were not designed by us, but acquired from other sources (including those that search for reductions), this is remarkable. 

\section{Related Work}\label{sec:related-work}


The most closely related work to ours is a line of work that identifies the gap in {\em relational verification} with the aid of CHC solving \cite{BjornerGMR2015,DeAngelisFPP2016,AngelisFPP2018a,FaisalAlAmeenKS2024,MordvinovF2017,MordvinovF2019} where the models are expressive enough to encode arbitrary recursion schemes like ours. While we operate at the recursive program level, these techniques directly target the CHC encodings \cite{BjornerGMR2015} of these programs and employ transformations such as fold/unfold \cite{DeAngelisFPP2016,AngelisFPP2018a}. In \cite{FaisalAlAmeenKS2024}, asynchronous fold/unfold transformation is proposed for fixpoint logic; the transformed logic formula is encoded and passed to a CHC solver.
Our product programs go beyond a bounded number of unfoldings. Recall the \ref{p3} property of $\mathtt{div}$. In the spirit of fold/unfold, one may think that unfolding once each of the $\mathtt{div}(n,d)$ and $\mathtt{div}(n',d)$ instances and twice the $\mathtt{div}(n'',d)$ instance may yield a product program, but this would not work due to a problem with the base case. Our product program conceptually unfolds $\mathtt{div}(n+n', d)$ for $n'$ times, which is essential in handling cases like this.

{\sc CHCproduct} \cite{MordvinovF2017} is a family of transformations that resemble synchronous lockstep. The net effect is not theoretically equivalent to our Figure~\ref{fig:ls} with $\vec{s} = (1, \ldots, 1)$, since {\sc CHCproduct} operates on the symbolic encoding of the programs, which allows the function calls to be ``rearranged'' and ``duplicated''. This flexibility also means the {\sc CHCproduct} is not uniquely defined, and the extension of \emph{synchronous} {\sc CHCproduct}~\cite{MordvinovF2017} employs a heuristic to select a unique  product. In \cite{MordvinovF2019}, an alternative technique based on Property Directed Reachability is proposed, which maintains over- and under-approximations of groups of predicates and automatically infers synchronous lockstep strategies, also on the symbolic encoding of the programs.

Putting CHC aside, in \cite{EilersMH2018} for programs with procedures, \emph{modular product programs} are constructed via a source-to-source transformation. Such product programs are intuitively a $(1,\ldots,1)$-lockstep reduction of copies of the same program. In this special case, modular product programs introduce Boolean \emph{activation variables} to encode control-flow symbolically, thereby avoiding taking the product of control locations.

\paragraph{\bf Multi-Stack Models and Sequentialization}

The languages of runs of a multi-stack pushdown automata are equivalent to those of Turing machines, even with two stacks. 
There are a number of solutions in the literature that deal with this problem by  under-approximating these multi-stack languages. The most well-known ones are based on bounding a measure in the set of runs; these include  bounding the number of context switches \cite{TorreMP2007,QadeerR2005}, phase bounding \cite{TorreMP2007,LaTorreMP2008}, and scope bounding \cite{LaTorreN2011,LaTorreNP2014a,LaTorreNP2020}. 
Our reductions are not an under-approximation, but rather a faithful representation. The languages are also theoretically incomparable against all these models. For example, in a bounded context switching scheme, one is permitted to break the respective order of calls and returns a bounded number of times, whereas it cannot be done in our model. In contrast, we include runs that are not bounded in up to any of the measures: number of context switches, phases, or scopes. Our focus in this paper is not on developing new classes of languages or giving new decidability results.  There is a loose connection between reductions and notion of {\em sequentialization} \cite{LalR2008} in the literature based on the same ideas of bounding the number of context switches, phases, or scopes. The well-nested shuffle can be considered as sequentialization that is {\em complete} under the conditions given in Section \ref{sec:concurrency}. It would be interesting to explore in future how well {\em incomplete} reductions of concurrent programs perform as means of under-approximations for bug finding.

\paragraph{\bf Hyperproperty Verification}
There is a large body of work studying verification of hyperproperties, where earlier work mostly focused on verification of hypersafety properties \cite{AntonopoulosKLN+2023,BartheCK2011,BartheDR2004,EilersMH2018,FarzanV2019,GodlinS2013,ShemerGSV2019,SousaD2016,TerauchiA2005,YangVSG+2018,ItzhakySV2024,UnnoTK2021}. Since our focus is on verification of infinite-state hypersafety properties, we forgo  surveying the verification of hyperliveness properties for finite-state programs, unless there is an overlap in techniques that makes the comparison worthwhile. 
There is a cluster of work \cite{BozzelliPS2021,GutsfeldMO2024,AgrawalB2016} on verification of hyperproperties \cite{ClarksonFKM+2014} based on \mbox{HyperLTL} (a temporal logic for hyperproperties) \cite{ClarksonFKM+2014}.
The work \cite{GutsfeldMO2024} on the verification of finite-state recursive programs, in the context of asynchronous HyperLTL,  is the only one from the cluster that addresses recursion in the hyperproperty space. The focus of \cite{GutsfeldMO2024} is on the decidability results of the relaxations of the logic for finite state recursive programs. We are in the scope of infinite-state programs, where verification is undecidable even in the absence of recursion. The state space of alignments (i.e., the degree of asynchrony) in \cite{GutsfeldMO2024} is not as general as an arbitrary reduction we define in Section~ \ref{sec:reduction}.


\paragraph{\bf Alignments for Iterative Programs}

Putting recursion aside, the closest to our work in terms of aim and scope are those that focus on verification of infinite-state programs and take advantages of the state space of alignments to solve the verification task. Specifically, in \cite{ShemerGSV2019,SousaD2016,FarzanV2019}, automated hypersafety verification techniques for sequential and concurrent programs \cite{FarzanV2019} are proposed. In \cite{UnnoTK2021}, a generalization of CHCs with a specialized solving algorithm is introduced, which is able to encode hyperproperties beyond hypersafety. In \cite{BeutnerF2022}, an automated game-based verification technique is used for a subset of hyperproperties that includes hypersafety properties but go beyond to include some classes of hyperliveness. Finally, in \cite{ItzhakySV2024}, a solution based on CHC encoding of a similar class of hyperproperties is proposed, which searches for a semantic alignment and a proof simultaneously. The tool from \cite{ItzhakySV2024} can handle some infinite-state programs, but for others, it requires the right {\em abstraction}, given by the user, that would translate the program into a finite-state one; consequently, the search space for alignments also become finite. \ruotong{TODO: I should check whether they did use a full predicate abstraction. If no, I should soften this.} As an instance, none of the three properties of the iterative version of the $\mathtt{div}$ example (from Section \ref{sec:intro}) can be verified by the tool over the concrete state. In some sense, recursive programs inherit all the complications of the iterative framework and add an extra layer of complexity. We discuss this extra complexity in Section \ref{sec:intro} and also note that the existence of a stack means even a program with finite data domain may have infinite state space.

In the paper, we mostly target the additional layer of complexity, and do not offer a new insight for iterative programs, with one exception: \emph{canonical reductions as building blocks} for the user to communicate reduction ideas to a verification tool. More precisely, the analogous set of lex reductions for iterative programs were already introduced in \cite{FarzanKP2022}, but the analog of canonical reductions (as a simple and effective way to specify a reduction) does not exist in the iterative program literature, to the best of our knowledge. It is noteworthy that the iterative counterparts of the canonical reductions are far easier to define formally and implement. Therefore, they may offer a practical solution to program alignments that sits between one fixed reduction (which has limited expressivity) and automated alignment search (which has limited scalability \cite{FarzanV2019} or requires user-provided abstractions in practice \cite{ItzhakySV2024}).





\section*{Data-Availability Statement}
Our tool {\tool} and the benchmarks used in the evaluation (Section \ref{sec:experiments}) can be found at \cite{ChengF2025a}. A description of all benchmarks is provided in \refapp{app:benchmarks}.

\bibliographystyle{ACM-Reference-Format}
\bibliography{biblio,anon}

\iftoggle{ext}{%
\appendix
\section{Background} \label{app:background}

\subsection{Recursive Programs as VPLs} \label{app:recursive-programs}

Following \cite{HeizmannHP2010}, we model a (single-stack) recursive program as a VPL (a well-matched one, for our purpose), where the alphabet $\tS$ is the program statements. Such VPLs can be obtained from a recursive control-flow graph by viewing the graph as a VPA. Not all words in the language is a semantically valid run of the program: a sequence of statements (also called a \emph{trace}) that repsects the control-flow structure and correctly matches the calls and returns may be \emph{infeasible} when semantics is taken into account. Below, we define a concrete alphabet $\tS$ and its semantics. We use it in the examples in Appendix~\ref{app:examples} and in our implementation of {\tool}. The theory developed in Section~\ref{sec:reduction}--\ref{sec:semantics} does not rely on this concrete alphabet.

For any set $X$, let $\stack(X)$ denote the set of stacks whose frames are elements of $X$, and let $\powerset(X)$ denote the power set of $X$. Let $L$ be the set of local variables, $F$ the set of function names, $L_{\mathsf{p}} \subset L$ the set of function parameters, $\param: F \to \powerset(L_{\mathsf{p}})$ that maps each function name to a set of parameters, and $\fout: F \to \powerset(L)$ the maps each function name to a set of output variables. Define $\bar{L}_{\mathsf{p}} := \{ \bar{p} \mid p \in L_{\mathsf{p}} \}$, which represents dummy variables that record the initial value of the parameters when a function is called. Let $D$ be the domain of the variables and $T  = D^{L \cup \bar{L}_{\mathsf{p}}}$ the set of all valuations.

We assume that $\tS$ consists of \emph{assign} statements \letter{$x$ := $t$}, \emph{assume} statements \letter{assume $\phi$}, \emph{call} statements $\letter{call $f$}$, and \emph{return} statements \letter{ret $f$}, where $x \in L$, $t$ is an expression over $L$, $\phi$ is an assertion over $L$, and $f \in F$. The first two types of statements are internal. For a program $P$ over this alphabet $\tS$, we always assume that for any word in $P$, if a call \letter{call $f$} and a return \letter{ret $f$} match, then $f = g$. Define the semantics function $\mathcal{S}: \tS \to \powerset(\stack(T) \times \stack(T))$ as follows:
\begin{align*}
	\sem{S}{x := t} & := \{ (S.\nu,\, S.\nu') : \nu' = \nu \oplus \{ x \mapsto \nu(t) \} \}, \\
	\sem{S}{\phi} & := \{ (S.\nu,\, S.\nu) : \nu \models \phi \}, \\ 
	\sem{S}{\call\; f} & := \{ (S.\nu,\, S.\nu.\nu') : \forall p\in \param(f).\, \nu'(p) = \nu'(\bar{p}) = \nu(p) \}, \\
	\sem{S}{\ret\; f} & := \{ (S.\nu_<.\nu,\, S.\nu') : \nu' = \nu_< \oplus \{ y \mapsto \nu(y) : y \in \fout(f) \} \}.
\end{align*}
Extend $\mathcal{S}$ from letters to words via relation composition. For a program $P$ over $\tS$, a run $\rho \in P$ is said to be infeasible if $\sem{S}{\rho} = \varnothing$. When talking about the runs of a program, we by default include both the feasible ones and infeasible ones.

\section{Examples of Product Programs} \label{app:examples}

The formal semantics of the alphabet that we use is presented in Appendix \ref{app:recursive-programs}.

\subsection{$\mathtt{div}$}

The grammar for $\mathtt{div}$ is as follows:

\begin{align*}
    D &\rightarrow \letter{assume $n < d$}\; \letter{$q$ := $0$} \\
    D &\rightarrow \letter{assume $n \ge d$}\; \letter{call $\mathtt{div}$}\; D\; \letter{ret $\mathtt{div}$}\; \letter{$q$ := $q + 1$}
\end{align*}

The grammar for synchronously lockstepping two copies of $\letter{call $\mathtt{div}$}\; D\; \letter{ret $\mathtt{div}$}$ is as follows, where \letter{$\sigma$ : $k$} means performing $\sigma$ on the $k$-th stack and $S$ is the start symbol:

\begin{align*}
    S {}\rightarrow{}& \letter{call $\mathtt{div}$ : 1}\; A\; \letter{ret $\mathtt{div}$ : 1} \\
    A {}\rightarrow{}& \letter{assume $n < d$ : 1}\; D_2 \\
    A {}\rightarrow{}& \letter{assume $n \ge d$ : 1}\; \letter{call $\mathtt{div}$ : 2}\; B\; \letter{ret $\mathtt{div}$ : 2}\; \letter{$q$ := $q + 1$ : 1} \\
    B {}\rightarrow{}& \letter{assume $n < d$ : 2}\; D_1 \\
    B {}\rightarrow{}& \letter{assume $n \ge d$ : 2}\; \letter{call $\mathtt{div}$ : 1}\; A\; \letter{ret $\mathtt{div}$ : 1}\; \letter{$q$ := $q + 1$ : 2} \\
    D_1 {}\rightarrow{}& \letter{assume $n < d$ : 1}\; \letter{$q$ := $0$ : 1} \\
    D_1 {}\rightarrow{}& \letter{assume $n \ge d$ : 1}\; \letter{call $\mathtt{div}$ : 1}\; D_1\; \letter{ret $\mathtt{div}$ : 1}\; \letter{$q$ := $q + 1$ : 1} \\
    D_2 {}\rightarrow{}& \letter{assume $n < d$ : 2}\; \letter{$q$ := $0$ : 2} \\
    D_2 {}\rightarrow{}& \letter{assume $n \ge d$ : 2}\; \letter{call $\mathtt{div}$ : 2}\; D_2\; \letter{ret $\mathtt{div}$ : 2}\; \letter{$q$ := $q + 1$ : 2}
\end{align*}

This does not exactly match the lockstep that we formally define when it comes to the ordering of internals but has the same spirits.

\section{Proofs} \label{app:proofs}
\begin{proof}[Proof of Proposition \ref{prop:uncountable}]
    Let $\Sigma = \{ a, b\}$, $P_1 = aa^*$, and $P_2 = bb^*$. Each minimal reduction corresponds to a choice function on $(P_1 \shuffle P_2) / {\equiv_{\Ith}}$. Note that $(P_1 \shuffle P_2) / {\equiv_{\Ith}}$ is an infinite set, and each element in the quotient has size at least $2$.
\end{proof}

\begin{proof}[Proof of Proposition \ref{prop:total-clo}]
    Let $\preceq$ be a CLO induced by a contextual order $\prec$.
    
    To show transitivity, suppose $\alpha \preceq \beta$ and $\beta \preceq \gamma$. Suppose $\alpha = xay$, $\beta = xbz = x'b'z'$, and $\gamma = x'ct$, where $a \prec_x b$ and $b' \prec_{x'} c$.
    \begin{itemize}
        \item If $x = x'$, then $\gamma = xct$ and $a \prec_x c$, so $\alpha \preceq \gamma$.
        \item If $x$ is a proper prefix of $x'$, then $xb$ is a prefix of $\gamma$ and $a \prec_x b$, so $\alpha \preceq \gamma$.
        \item If $x'$ is a proper prefix of $x$, then $x'b'$ is a prefix of $\alpha$ and $b' \prec_{x'} c'$, so $\alpha \preceq \gamma$.
    \end{itemize}
    The three other cases can be verified in the same way, and transitivity holds. Similarly, we can verify reflxivity and antisymmetry.

    For the sufficient condition for totality, for any distinct two distinct words, consider the first position where they differ.
\end{proof}

\begin{proof}[Proof of Theorem \ref{thm:VPLs-closed-under-wnshuffle}]
    Let $A_{i}=(Q_{i},Q_{i}^{\mathsf{in}},\Gamma_{i},\delta_{i},Q_{i}^{\mathsf{F}})$ be a VPA over the visibly pushdown alphabet $\tS_{i}$. Define the VPA $A=(Q,Q^{\mathsf{in}},\Gamma,\delta,Q^{\mathsf{F}})$
    over $\tS$ as follows:

    \begin{itemize}
        \item $Q=Q_1\times \cdots \times Q_n$.
        \item $Q^{\mathsf{in}}=Q_1^{\mathsf{in}}\times \cdots \times Q_n^{\mathsf{in}}$.
        \item $\Gamma=(\Gamma_1\setminus\{\bot\})\uplus \cdots \uplus (\Gamma_2\setminus\{\bot\})\uplus\{\bot\}$.
        \item $\delta^{\mathsf{int}} = \{(\vec{q}, \sigma, \vec{q}[i \mapsto q_i']) : \vec{q} \in Q,\, i \in \{1, \ldots, n\},\, (q_i, \sigma, q_i') \in \delta_i \}$.
        \item $\delta^{\mathsf{call}} = \{(\vec{q}, \sigma, \vec{q}[i \mapsto q_i'], \gamma) : \vec{q} \in Q,\, i \in \{1, \ldots, n\},\, (q_i, \sigma, q_i', \gamma) \in \delta_i \}$.
        \item $\delta^{\mathsf{ret}} = \{(\vec{q}, \sigma, \gamma, \vec{q}[i \mapsto q_i']) : \vec{q} \in Q,\, i \in \{1, \ldots, n\},\, (q_i, \sigma, \gamma, q_i') \in \delta_i \}$.
        \item $Q^{\mathsf{F}}=Q_1^{\mathsf{F}} \times \cdots \times Q_n^{\mathsf{F}}$.
    \end{itemize}
    It can be verified that $\Lang(A)=\Lang(A_1)\wnshuffle\cdots\wnshuffle\Lang(A_n)$.
\end{proof}

\begin{proof}[Proof of Proposition \ref{prop:vp-sigma-star}]
    We prove the statement for a generic $I$ using the sleep set \cite{Godefroid1996} construction. Suppose $\prec$ is represented by a complete and deterministic VPA $(Q_S, Q_S^{\mathsf{in}}, \Gamma_S, \delta_S, Q_S)$. Define the VPA $A = (Q, Q^{\mathsf{in}}, \Gamma, \delta, Q^{\mathsf{F}})$ as follows:
    \begin{itemize}
        \item $Q = 2^{\tS} \times Q_S$.
        \item $Q^{\mathsf{in}} = \{\varnothing\} \times Q_S^{\mathsf{in}}$.
        \item $\Gamma = \Gamma_S$.
        \item $\delta$ is the smallest subset of $(Q\times\Sc\times Q\times(\Gamma\setminus\{\bot\}))\cup(Q\times\Sr\times\Gamma\times Q)\cup(Q\times\Sl\times Q)$ such that for all set $r \in 2^{\tS}$, states $s, s' \in Q_S$, letter $\sigma \in \tS \setminus r$, and stack symbol $\gamma \in \Gamma$, writing $r' = \{ \sigma' \in \tS \mid \sigma' \in I(\sigma) \land (((\sigma', \sigma) \in \mathsf{ord}(s)) \lor (\sigma' \in r)) \}$,
        \begin{itemize}
            \item if $(s, \sigma, s', \gamma) \in \delta_S$, then $((r, s), \sigma, (r', s'), \gamma) \in \delta$;
            \item if $(s, \sigma, \gamma, s') \in \delta_S$, then $((r, s), \sigma, \gamma, (r', s')) \in \delta$; and
            \item if $(s, \sigma, s') \in \delta_S$, then $((r, s), \sigma, (r', s')) \in \delta$.
        \end{itemize}
        \item $Q^{\mathsf{F}} = Q$.
    \end{itemize}
    It can be shown that $\Lang(A) = \red_\prec(\tS^*)$.
\end{proof}

\begin{proof}[Proof of Proposition \ref{prop:nonvpllex} and Proposition \ref{prop:dummy}]
    Let $\tS_1,\tS_2$ be the visibly pushdown alphabet of parentheses and brackets respectively. Define $\prec:\Sigma^{*}\to\Lin(\Sigma)$ that maps strings of an even length to $\texttt{(}<\texttt{)}<\texttt{[}<\texttt{]}$, and strings of an odd length to $\texttt{[}<\texttt{]}<\texttt{(}<\texttt{)}$. Let $L_1$ be the language of balanced parentheses and $L_2$ be the language of balanced brackets.

    Suppose $\red_\prec(L_1\shuffle L_2)$ is context-free, and let $p$ be a pumping length. Pumping the string 
    \[
        \underbrace{\texttt{([}\cdots\texttt{([}}_{p\text{ times}}\underbrace{\texttt{(]}\cdots\texttt{(]}}_{p\text{ times}}\texttt{)}^{2p}
    \]
    leads to a contradiction, so $\red_\prec(L_1\shuffle L_2)$ cannot be context-free.
\end{proof}

\begin{proof}[Proof of Theorem \ref{thm:vpwnred}]
    If $L$ is closed under $I$ then  
    \[
    \red_{I,\prec}(L)=\red_{I,\prec}(\tS^{*})\cap L.
    \]
    
    Suppose $\prec$ is a VP contextual order and $\red_{I,\prec}(P_1 \shuffle P_2) \subset P_1 \wnshuffle P_2 \subset P_1 \shuffle P_2$. It follows that
    \[
        \red_{I,\prec}(P_1 \shuffle P_2) = \red_{I,\prec}(P_1 \shuffle P_2) \cap (P_1 \wnshuffle P_2) = \red_{I,\prec}(\tS^*) \cap (P_1 \wnshuffle P_2).
    \]
    By Proposition \ref{prop:vp-sigma-star}, Theorem \ref{thm:VPLs-closed-under-wnshuffle}, and the closedness of VPLs under intersection, it follows that $\red_{I,\prec}(P_1 \shuffle P_2)$ is visibly pushdown.
\end{proof}

\begin{proof}[Proof of Proposition \ref{prop:coherency-implies-well-nestedness}]
    Write $\tS_i = \left\langle \Sc_i, \Sr_i, \Sl_i \right\rangle$ ($i = 1, 2$). Suppose for contradiction there is some word in $\red_\prec(L_1 \shuffle L_2) \setminus (L_1 \shuffle L_2)$. Use $\texttt{(}$, $\texttt{)}$, $\texttt{[}$, and $\texttt{]}$ to denote a letter in $\Sc_1$, $\Sr_1$, $\Sc_2$, and $\Sr_2$ respectively. Consider the shortest prefix where the well-nestedness is broken, and WLOG suppose the last letter in the prefix is a $\mathtt{]}$. Then $\mathtt{]}$ is wrongly matched with a $\mathtt{(}$. Then the word has the form $\alpha \mathtt{(} \beta \mathtt{]} \gamma \mathtt{)} \delta$. Note that the $\mathtt{(}$ is the last pending call in the prefix $\alpha \mathtt{(} \beta$. The first letter from $\tS_1$ in $\gamma \mathtt{)}$ can be swapped to the position of the $\mathtt{]}$ to give a smaller interleaving, which is a contradiction.
\end{proof}

\begin{proof}[Proof of Theorem \ref{thm:vpred-vp-coherent}]
    By Proposition~\ref{prop:coherency-implies-well-nestedness} and Theorem~\ref{thm:vpwnred}.
\end{proof}

\begin{proof}[Proof of Lemma \ref{lem:repair}]
    The coherency of $\prec'$ is evident. The VPA for $\prec'$ can be constructed from the VPA for $\prec$ and an additional state to track (1) whether there is a pending call and (2) the component which the last pending call is from.

    Any equivalence class $E$ of $P_{\wnshuffle} = P_1 \wnshuffle \cdots \wnshuffle P_n$ is the restriction of a unique equivalence class $E'$ of $P_{\shuffle} = P_1 \shuffle \cdots \shuffle P_n$ to $P_{\wnshuffle}$. The definition of $\prec'$ ensures that the $\min_{\preceq}(E) = \min_{\preceq'}(E')$. Then, since any equivalence class of $P_{\shuffle}$ contains a unique equivalence class of $P_{\wnshuffle}$, it follows that $\red_\prec(P_{\wnshuffle}) = \red_{\prec'}(P_{\shuffle})$.
\end{proof}

\begin{proof}[Proof of Theorem \ref{thm:vpred}]
    Suppose $\prec$ is a VPO and $P_1, \ldots, P_n$ are VPLs (well-matched by our global assumption). By Lemma \ref{lem:repair}, for some coherent VPO $\prec'$, we have
    \[
        \red_\prec(P_1 \wnshuffle \cdots \wnshuffle P_n) = \red_{\prec'}(P_1 \shuffle \cdots \shuffle P_n).
    \]
    Then, by Theorem \ref{thm:vpred-vp-coherent}, $\red_\prec(P_1 \wnshuffle \cdots \wnshuffle P_n)$ is visibly pushdown.
\end{proof}

\begin{proof}[Proof of Theorem \ref{thm:oc}]
    Suppose $A = (Q_A, \{q_A^{\mathsf{in}}\}, \Gamma_A, Q_A, \delta_A)$ and $P_i$ is given by $A_i = (Q_i, Q_i^{\mathsf{in}}, \Gamma_i, F_i, \delta_i)$.
    Suppose $\prec$ is uniform w.r.t. $A_1, A_2$. We further assume that $F_i$ has no outgoing transition; $A_i$ can be modified to satisfy this condition without breaking uniformity. Assume $\Gamma_1 \cap \Gamma_2 = \{\bot\}$ and let $\Gamma_{12} = \Gamma_1 \cup \Gamma_2$. For any $q_1 \in Q_1$, $q_2 \in Q_2$, and $\gamma \in \Gamma_{12}$, define $\enabled(q_1, q_2, \gamma)$ as the set of outgoing letters enabled at $q_1$ or $q_2$ when the stack top is $\gamma$ (if $\gamma \notin \Gamma_i$, then $\enabled(q_1, q_2, \gamma)$ contains no letter from $\Sr_i$).

    We construct $R = (Q, Q^{\mathsf{in}}, \Gamma, F, \delta)$, which simulates the interleaving of $A_1$ and $A_2$, with a conceptual stack $S$, guided by $A$. 
    \begin{itemize}
        \item $Q = Q_1 \times Q_2 \times Q_A \times \Gamma_{12}$. The state $(q_1, q_2, q_A, \gamma) \in Q$ maintains the states $q_1, q_2, q_A$ in $A_1, A_2, A$ respectively and the top $\gamma$ of the stack $S$.
        \item $Q^{\mathsf{in}} = Q_1^{\mathsf{in}} \times Q_2^{\mathsf{in}} \times \{ q_A^{\mathsf{in}}\} \times \{ \bot \}$.
        \item $\Gamma = (\Gamma_{12} \times \Gamma_A \times \Gamma_{12}) \cup \{\bot\}$. When $(\gamma, \gamma_A, \gamma')$ is the stack top of $R$, the top two elements of $S$ are $\gamma$ and $\gamma'$ (in that order), and the stack top of $A$ is $\gamma_A$.
        \item $F = F_1 \times F_2 \times Q_A \times \Gamma_{12}$.
        \item 
        
        For any $q = (q_1, q_2, q_A, \gamma) \in Q$, if $\min_{\ord(q_A)}(\enabled(q_1, q_2, \gamma)) \in \tS_1$,
        \begin{itemize}
            \item for all $(q_1, \sigma, q'_1) \in \dl_1$ and $(q_A, \sigma, q'_A) \in \dl_A$, we have
            \[ (q, \sigma, (q'_1, q_2, q'_A, \gamma)) \in \dl; \]
            \item for all $(q_1, \sigma, q'_1, \gamma_1) \in \dc_1$ and $(q_A, \sigma, q'_A, \gamma_A) \in \dc_A$, we have
            \[ (q, \sigma, (q'_1, q_2, q'_A, \gamma_1), (\gamma_1, \gamma_A, \gamma)) \in \dc; \]
            \item for all $(q_1, \sigma, \gamma, q'_1) \in \dr_1$, $(q_A, \sigma, \gamma_A, q'_A) \in \dr_A$, and $\gamma' \in \Gamma_{12}$, we have
            \[ (q, \sigma, (\gamma, \gamma_A, \gamma'), (q'_1, q_2, q'_A, \gamma')) \in \dr. \]
        \end{itemize}
        The case where $\min_{\ord(q_A)}(\enabled(q_1, q_2, \gamma)) \in \tS_2$ is symmetrical, and $\delta$ is the smallest set defined by these rules.
    \end{itemize}
    Let $P$ be the automaton for $P_1 \wnshuffle P_2$ constructed in Theorem \ref{thm:VPLs-closed-under-wnshuffle}. By relating accepted runs of $R$ to those of $P$, one can show that $\Lang(R) \subset \Lang(P)$. Furthermore, for any word $\rho$ over $\tS$, $\rho \in \Lang(R) \Leftrightarrow \rho \in \red_\prec(P_1 \wnshuffle P_2)$. We omit a formal proof by induction but point out the key lemmas.
    \begin{itemize}
        \item[$\Rightarrow$:] For any word $\rho \in \Lang(R)$ and words $\alpha, \beta$ such that $\rho = \alpha \beta$, if $R$ reaches state $(q_1, q_2, q_A, \gamma)$ after reading $\alpha$ in an accepted run over $\rho$, then $x \in \enabled(q_1, q_2, \gamma)$ for any letter $x \in \letters(\beta)$ that can commute to the right of $\alpha$.
        \item[$\Leftarrow$:] For any word $\rho \in \red_\prec(P_1 \wnshuffle P_2)$ and words $\alpha, \beta$ and letter $x \in \tS_i$ such that $\rho = \alpha x \beta$, if $P$ reaches $(q_1, q_2)$ with stack top $\gamma$ after reading $\alpha$ in an accepted run over $\rho$, then $x \prec_\alpha y$ for all letter $y \in \enabled(q_1, q_2, \gamma) \setminus \tS_i$.
        
        Note that this does not hold without the assumption that $F_1$ and $F_2$ have no outgoing transition.
    \end{itemize}

    Lastly, we observe that $|Q| = |Q_1||Q_2||Q_A||\Gamma_{12}|$.
\end{proof}

For any $\tS_i$, let $\tS_i^\dep$ be all letters in $\tS_i$ that do not soundly commute with all letters in other components. In other words, it is the complement of $\tS_i^\indep$.

\begin{lemma} \label{lem:tail-indep-grammar}
    Let the nonterminal $Y$ generate exactly the well-matched runs over $\tS_i^\indep$. The following statements are equivalent:
    \begin{enumerate}
        \item $P_i$ is tail-independent (resp. head-independent).
        \item For every run $\rho$ of $P_i$ and every pair of positions $j < j'$ (resp. $j > j'$) in $\rho$ such that $\rho_j, \rho_{j'} \in \tS_i^{\mathsf{dep}}$, for all positions $k, l$ in $\rho$ such that $k \rightsquigarrow l$, we have $k < j < l \Rightarrow k < j' < l$.
        \item Every run of $P_i$ is derivable from $X$, whose grammar is given by
        \[ X \to Y \mid aX \mid cYrX \mid cXrY, \]
        where $a$, $c$, and $r$ range over all letters in $\Sl_i$, $\Sc_i$, and $\Sr_i$ respectively. (resp. $X \to Y \mid Xa \mid XcYr \mid YcXr$)
    \end{enumerate}
\end{lemma}
\begin{proof}
    We prove the equivalence for tail-independence. The case for head-independence is symmetrical.
    \begin{description}
        \item[(1) $\Rightarrow$ (2):] Suppose (1). Let $\rho$ be a run of $P_i$. Write $\rho = uv$ as in the definition. Let $j < j'$ be a pair of positions such that $\rho_j, \rho_{j'} \in \tS_i^{\mathsf{dep}}$. Then they are both in $u$. Let $k, l$ be positions such that $k \rightsquigarrow l$ and $k < j < l$. Then $k$ is in $u$ and $l$ is in $v$, which shows that $k < j < j' < l$.

        \item[(2) $\Rightarrow$ (1):] Suppose (2). Let $\rho$ be a run of $P_i$. If there is no letter in $\tS_i^\dep$ in $P_i$, then any division $\rho = uv$ satisfies the requirement. Otherwise, let $u$ be the shortest prefix of $\rho$ that contains the last letter in $\tS_i^\dep$ and write $\rho = uv$. Then all letters in $v$ are in $\tS_i^\indep$.
        
        Suppose $k \rightsquigarrow l$ and there is some letter in $\tS_i^\dep$. It suffices to show that $k$ is in $u$ and $l$ is in $v$. Indeed: since any letter in $\tS_i^\dep$ must be in $u$, the position $k$ is in $u$; since $l$ is behind the last letter in $\tS_i^\dep$, the position $l$ is in $v$.

        \item[(2) $\Rightarrow$ (3):] We will prove the stronger claim that for any well-matched run over $\tS_i$, if for every pair of positions $j < j'$ in $\rho$ such that $\rho_j, \rho_{j'} \in \tS_i^{\mathsf{dep}}$, for all positions $k, l$ in $\rho$ such that $k \rightsquigarrow l$, we have $k < j < l \Rightarrow k < j' < l$, then $\rho$ is derivable from $X$.
        
        Let $\rho$ be well-matched run over $\tS_i$ such that the condition holds. Assume by induction that the statement holds for shorter runs.
        
        If $\rho = \epsilon$, then $\rho$ is derivable from $Y$ and therefore $X$.

        If $\rho = av$, where $a \in \Sl_i$, then $v$ also satisfies the condition, no matter whether $a \in \tS_i^\dep$ or not, and is therefore derivable from $X$. Thus, $\rho$ is derivable from $X$.

        If $\rho = cwrv$, where $c \in \Sl_i$ and $r$ is the matching return, then $w$ also satisfies the condition and is thus derivable from $X$; all letters in $v$ must be in $\tS_i^\indep$, so $v$ is derivable from $Y$. Hence, $X \Rightarrow cXrY \Rightarrow^* \rho$.

        \item[(3) $\Rightarrow$ (2):] Suppose (3). Let $\rho$ be a run of $P_i$. Suppose words shorter than $\rho$ all satisfy the requirement in (2). Let $j < j'$ be a pair of positions such that $\rho_j, \rho_{j'} \in \tS_i^{\mathsf{dep}}$. Consider the first step of a derivation of $\rho$ from $X$. We will show that for all positions $k, l$ in $\rho$ such that $k \rightsquigarrow l$, we have $k < j < l \Rightarrow k < j' < l$.
        \begin{description}
            \item[$X \Rightarrow Y$.] This cannot happen due to the existence of $\rho_j$.
            \item[$X \Rightarrow aX$.] If $j$ is the position of $a$, then no $k, l$ exist. Otherwise, by induction hypothesis.
            \item[$X \Rightarrow cYrX$.] Both $j$ and $j'$ can only be in $X$ on the RHS. Use induction hypothesis.
            \item[$X \Rightarrow cXrY$.] Still, both $j$ and $j'$ can only be in $X$ on the RHS. Suppose $k \rightsquigarrow l$ and $k < j < l$. If $k, l$ are the positions of $c, r$ respectively, then they also enclose $j'$. Otherwise, they are inside $X$ on the RHS, and we use induction hypothesis.
        \end{description}
    \end{description}
\end{proof}

\begin{proof}[Proof of Proposition \ref{prop:tail-indep-vpl}]
    By Lemma \ref{lem:tail-indep-grammar}, $P_i$ is tail-independent iff $P_i \subset \Lang(X)$. Note that
    \[ Y \to \epsilon \mid aY \mid cYrY, \]
    where $a$, $c$, and $r$ range over all letters in $\Sl_i \cap \tS_i^\indep$, $\Sc_i$, and $\Sr_i$ respectively. (Recall we assume that calls and returns fully commute with letters in other components.) Thus, the grammar for $X$ can be rewritten to a VPG, which shows $\Lang(X)$ is a VPL.

    We have $P_i \subset \Lang(X) \iff P_i \cap \stcomp{\Lang(X)} = \varnothing$. By \cite[Theorem 6]{AlurM2004}, one can convert the VPG for $\Lang(X)$ into a VPA with $O(1)$ states and stack alphabet of size $O(|\tS_i|)$. By \cite[Theorem 2]{AlurM2004}, this VPA can be determinized to have $O(1)$ states and stack alphabet of size $O(|\tS_i|)$. Then, $P_i \cap \stcomp{\Lang(X)}$ can be constructed as a VPA with polynomial size, and its emptiness can be decided in polynomial time.

    A similar argument, rewriting the grammar to a ``reversed'' VPG, holds for head-independence.
\end{proof}

\begin{proof}[Proof of Theorem \ref{thm:wnshuffle-sound}]
    We use the characterization (3) of tail-independence in Lemma \ref{lem:tail-indep-grammar}.

    Let $\pi_i$ be the projection $\Pi_{\tS_i}$. Let $\tS'_i$ be the subset of $\tS_i$ that soundly commutes with every letter in $\bigcup_{j\neq i} \tS_j$. Let $I$ be the largest sound commutativity relation.

    We will prove by induction that, if $P_1, \ldots, P_n$ are tail-independent singleton languages, then every word in $P_1 \shuffle \cdots \shuffle P_n$ is $I$-equivalent to some word in $P_1 \wnshuffle \cdots \wnshuffle P_n$.

    Let $P_1 = \{\rho_1\}, \ldots, P_n = \{\rho_n\}$ be tail-independent singleton languages and suppose the statement holds on smaller instances. Let $\rho \in P_1 \shuffle \cdots \shuffle P_n$. If $\rho$ contains no call (and thus no return), then the statement holds. Otherwise, write $\rho = \alpha c \beta r \gamma$, where $c$ is the first call in $\rho$, and $r \in \Sr_i$ matches $c$ in $\rho_i$.
    \begin{itemize}
        \item If $\letters(\pi_i(\beta)) \subset \tS'_i$, then $\rho \equiv_I \alpha c \beta' r \beta'' \gamma$, where $\beta' = \pi_i(\beta)$; in other words, we move letters in $\beta r$ that are in $\tS_i$ forward. Invoke the induction hypothesis on $(\{\pi_j(\beta''\gamma)\})_{j\in[n]}$.
        \item If $\letters(\pi_i(\beta)) \not\subset \tS'_i$, then the derivation of $\rho_i$ must have the form
        \[ X \Rightarrow a_1 X \Rightarrow \cdots \Rightarrow a_1\cdots a_n X \Rightarrow a_1\cdots a_n cXr\blue{Y} \Rightarrow \cdots \Rightarrow \rho_i, \]
        where $\blue{Y} \Rightarrow^* \pi_i(\gamma)$. Therefore, $\rho \equiv_I \alpha c \beta \gamma' r \gamma''$, where $\gamma'' = \pi_i(\gamma)$; in other words, we move letters in $r \gamma$ that are in $\tS_i$ backward. Invoke the induction hypothesis on $(\{\pi_j(\beta\gamma')\})_{j\in[n]}$.
    \end{itemize}

    For any languages $L_1, \ldots, L_n$, the shuffle (resp. well-nested shuffle) of $L_1, \ldots, L_n$ is the union of the shuffle (resp. well-nested shuffle) of $\{\rho_1\}, \ldots, \{\rho_n\}$ for all $\rho_1 \in L_1, \ldots, \rho_n \in L_n$. Therefore, if $P_1, \ldots, P_n$ are tail-independent languages that are not necessarily singletons, it still holds that every word in $P_1 \shuffle \cdots \shuffle P_n$ is $I$-equivalent to some word in $P_1 \wnshuffle \cdots \wnshuffle P_n$.

    The case for head-independent languages is symmetrical.
\end{proof}

\begin{proof}[Proof of Proposition \ref{prop:cr-vplex}]
    A VP contextual order is given by a complete deterministic VPA \\ $A = (Q, \{q^{\mathsf{in}}\}, \Gamma, Q, \delta)$ together with a map $\ord: Q \to \Lin(\Sigma)$. Below, we define $A$ and $\ord$ for each of concatenation, nested concatenation, and $\vec{s}$-lockstep. We assume $\Gamma = Q \uplus \{\bot\}$ and write the VPA like a weakly-hierarchical linear accepting NWA \cite{AlurM2009}, i.e., the automaton pushes $q$ if at state $q$ a call letter is read. Since $P_i$'s are well-matched, we omit the case when a pending return is read.
    \begin{description}
        \item[Concatenation] Let $A$ be a complete deterministic VPA with a single state $\star$. Let $\ord(\star)$ be any linear order such that $\Sigma_1 < \cdots < \Sigma_n$, i.e., if $x \in \Sigma_i$, $y \in \Sigma_j$, and $i < j$, then $x < y$.
        \item[Nested Concatenation] Let $A$ be a complete deterministic VPA with a single state $\star$. Let $\ord(\star)$ be any linear order such that $\Sl_1 \cup \Sc_1 < \cdots < \Sl_n \cup \Sc_n < \Sr$.
        \item[$\vec{s}$-lockstep] Let $Q = \{ (t_1, \ldots, t_n) \mid 0 \le t_i \le s_i \text{ for each } i \} \cup \{ \star \}$. Let $q^{\mathsf{in}} = (0, \ldots, 0)$. For all states $q, q', (t_1, \ldots, t_n) \in Q$ and letters $a_i \in \Sl_i,\, c_i \in \Sc_i,\, r_i \in \Sr_i$,
        \begin{align*}
        \dl(q, a_i) &= q, \\
        \dc(\star, c_i) &= \star, \\
        \dc((t_1, \ldots, t_n), c_i) &= \begin{cases}
            (0, \ldots, 0, s_i - 1, s_{i+1}, \ldots, s_n), &\text{if } t_i = 0\\
            (0, \ldots, 0, t_i - 1, t_{i+1}, \ldots, t_n), &\text{otherwise}
        \end{cases} \\
        \dr(q, q', r_i) &= \begin{cases}
            q', &\text{if } q' = (t_1, \ldots, t_n) \land t_i = 0 \\
            \star. &\text{otherwise} 
        \end{cases}
        \end{align*}
        $\ord(\star)$ is any linear order such that $\Sigma_1 < \cdots < \Sigma_n$. $\ord((t_1, \ldots, t_n))$ is any linear order such that $\Sl_1 < \cdots < \Sl_n < \Sc_j < \Sc_{j+1} < \cdots < \Sc_n < \Sc_1 < \cdots \Sc_{j-1} < \Sr$, where $j$ is smallest index such that $t_j > 0$.
    \end{description}
\end{proof}

\begin{proof}[Proof of Theorem \ref{thm:cr-vp}]
    By Proposition \ref{prop:cr-vplex}, induction, and the fact that for any permuation $\pi$ of $\{1, \ldots, n\}$,
    \[
        P_{\pi(1)} \wnshuffle \cdots \wnshuffle P_{\pi(n)} = P_1 \wnshuffle \cdots \wnshuffle P_n.
    \]
\end{proof}

\begin{proof}[Proof of Proposition \ref{prop:homo-vpo}]
    Suppose $\prec$ is represented by VPA $A = (Q_A,Q_A^{\mathsf{in}},\Gamma_A,\delta_A,Q_A^{\mathsf{F}})$ together with map $\ord$. Define VPA $A' = (Q_A,Q_A^{\mathsf{in}},\Gamma_A \uplus \{\star\},\delta_{A'},Q_A^{\mathsf{F}})$, where $\delta_{A'}$ is obtained from $\delta_A$ by
    \begin{itemize}
        \item replacing every $(q, c, q', \gamma) \in \dc_A$ such that $c \in \Sigma'$ with $(q, c, q'', \star)$, where $(q, f(c), q'') \in \delta_A$; and
        \item replacing every $(q, r, \gamma, q') \in \dr_A$ such that $r \in \Sigma'$ with $(q, r, \star, q'')$, where $(q, f(r), q'') \in \delta_A$.
    \end{itemize}

    Let $w = w_1\cdots w_k$ be any prefix of any run of $P_1 \wnshuffle \cdots \wnshuffle P_n$. Let $(q, s)$ be the configuration of $A$ after reading $f(w)$. Then the configuration of $A'$ after reading $w$ is $(q, s')$, where $s'$ is obtained from $s$ by inserting $\star$ pushed by all $w_i \in \Sigma'$ that is a pending call.

    Define map $\ord' : Q_A \to \Lin(\tS)$ over the states of $A'$ by mapping $q \in Q_A$ to any linear order extending $\{ (a, b) \mid (f(a), f(b)) \in \ord(q) \}$. One can show that such an extension always exists, e.g., by the Order-Extension Principle. We can take any extension because reducing $P_1 \wnshuffle \cdots \wnshuffle P_n$ under $\Ith$ only compares $a, b$ from different $\tS_i$'s, and in this case $f(a)$ and $f(b)$ are already comparable by $\ord(q)$. The contextual order $\prec'$ represented by VPA $A'$ together with map $\ord'$ satisfies $\red_{\prec'}(P_1 \wnshuffle \cdots \wnshuffle P_n) = \red_{f(\prec)}(P_1 \wnshuffle \cdots \wnshuffle P_n)$.
\end{proof}

\subsection{Semantics}
Let $\stktop(S)$ denote the top frame of a nonempty stack $S$ and let \[ \stktop(\vec{S}) := (\stktop(S_1), \ldots, \stktop(S_n))\] for a tuple of nonempty stacks $\vec{S} = (S_1, \ldots, S_n)$. We have the following lemma, which can be proved by induction. Note the right-hand side of the implication states the equation of two sets.

\begin{lemma} \label{lem:sem}
    For any $\rho \in \tS_1 \wnshuffle \cdots \wnshuffle \tS_n$ without pending returns, $\vec{S} \in \prod_{j=1}^n \stack(V_j)$, and $S \in \stack(V)$,
    \[
        \stktop(\vec{S}) = \stktop(S) \implies \stktop(\sem{M}{\rho}(\vec{S})) = \stktop(\sem{S}{\rho}(S)).
    \]
\end{lemma}

\begin{proof}[Proof of Theorem \ref{thm:sem}]
    Formally, for any $\rho \in \tS_1 \wnshuffle \cdots \wnshuffle \tS_n$ without pending returns,
    \begin{itemize}
        \item under the semantics $\mathcal{M}$, $\{A\} \rho \{B\}$ is valid iff $\stktop(\vec{S}) \models A \Rightarrow \stktop(\vec{S'}) \models B$ for all $(\vec{S}, \vec{S'}) \in \sem{M}{\rho}$;
        \item under the semantics $\mathcal{S}$, $\{A\} \rho \{B\}$ is valid iff $\stktop(S) \models A \Rightarrow \stktop(S') \models B$ for all $(S, S') \in \sem{S}{\rho}$.
    \end{itemize}

    Let $\rho \in P_1 \wnshuffle \cdots \wnshuffle P_n$. Suppose $\{A\} \rho \{B\}$ under the semantics $\mathcal{S}$. By Lemma \ref{lem:sem}, for any $(\vec{S}, \vec{S'}) \in \sem{M}{\rho}$, there is some $(S, S') \in \sem{S}{\rho}$ such that $\stktop(\vec{S}) = \stktop(S)$ and $\stktop(\vec{S'}) = \stktop(S')$; since $\{A\} \rho \{B\}$ under the semantics $\mathcal{S}$, we have $\stktop(S) \models A \Rightarrow \stktop(S') \models B$, and therefore $\stktop(\vec{S}) \models A \Rightarrow \stktop(\vec{S'}) \models B$. This shows that $\{A\} \rho \{B\}$ under the semantics $\mathcal{M'}$. The other direction is symmetrical.
\end{proof}

\section{Extra Details}
\subsection{CHC Encodings} \label{app:enc}

We use the alphabet introduced in Appendix \ref{app:recursive-programs}.

\paragraph{Directly Encode an NWA}

Let $A = (Q, Q^{\mathsf{in}}, Q^{\mathsf{F}}, \tS, \delta)$ be a linearly accepting, weakly hierarchical NWA without pending returns such that for any $(q, q', \letter{ret f}, q'') \in \delta$, any outgoing call letter at $q'$ is $\letter{call f}$. For each state $q \in Q$, declare a predicate $I_q(\widebar{\vec{x}}, \vec{l}, r)$, which represents an assertion that holds at $q$ relating the function arguments $\widebar{\vec{x}}$, the local variables $\vec{l}$, and the return address $r \in Q \uplus \{\bot\}$ at the current stack frame, where $\bot$ denotes the stack bottom. Each transition in $\delta$ translates to a constrained Horn clause:
\begin{itemize}
  \item $(q, a, q') \in \delta$ such that $a \in \Sl$ translates to
  \[
    I_q(\widebar{\vec{x}}, \vec{l}, r) \land \sem{S}{a}(\vec{l}, \vec{l}') \implies I_{q'}(\widebar{\vec{x}}, \vec{l}', r).
  \]
  \item $(q, \letter{call f}, q') \in \delta$ translates to
  \[
    I_q(\widebar{\vec{x}}, \vec{l}, r) \land \bigwedge_{p \in \param(f)} p' = \widebar{p}' = p \implies I_{q'}(\widebar{\vec{x}}', \vec{l}', q).
  \]
  \item $(q, q', \letter{ret f}, q'') \in \delta$ translates to
  \begin{multline*}
    I_q(\widebar{\vec{x}}, \vec{l}, q') \land I_{q'}(\widebar{\vec{x}}', \vec{l}', r) \land \left(\bigwedge_{p \in \param(f)} \widebar{p} = p'\right) \\ \land \left(\bigwedge_{l\in \vec{l}} l'' = \mathsf{ite}(l \in \fout(f),\, l,\, l')\right) \implies I_{q''}(\widebar{\vec{x}}', \vec{l}'', r).
  \end{multline*}
\end{itemize}

To verify a pair of pre/postconditions $\pre$ and $\post$, add the clauses
\begin{align*}
  \pre &\implies I_{q^{\mathsf{in}}}(\widebar{\vec{x}}, \vec{l}, \bot), \\
  I_{q^{\mathsf{F}}}(\widebar{\vec{x}}, \vec{l}, r) \land \neg \post &\implies \bot
\end{align*}
for all $q^{\mathsf{in}} \in Q^{\mathsf{in}}$, $q^{\mathsf{F}} \in Q^{\mathsf{F}}$, and $r \in Q \uplus \{\bot\}$. If it is known that $\Lang(A)$ is well-matched, then it suffices to take $r = \bot$.

The following proposition states the soundness and completeness of the encoding:

\begin{proposition}
  Let $A = (Q, Q^{\mathsf{in}}, Q^{\mathsf{F}}, \tS, \delta)$ be a linearly accepting, weakly hierarchical NWA without pending returns such that for any $(q, q', \letter{ret f}, q'') \in \delta$, any outgoing call letter at $q'$ is $\letter{call f}$. Let $P = \Lang(A)$.
  For the recursive program $(P, \mathcal{S})$, the CHC encoding is satisfiable iff $\{\pre\} P \{\post\}$.
\end{proposition}
\begin{proof}
  $\Rightarrow$: Suppose the CHC encoding is satisfiable. Let $\rho$ be a nested word with an accepting run $q_0 \xrightarrow{\rho_1} \cdots \xrightarrow{\rho_n} q_n$. Then $\rho$ can be annotated with an \emph{inductive sequence of nested interpolants}\footnote{In the context of this definition, an interpolant is a set of valuations. For our purpose, we remove the requirement that $J_0 = \top$ and $J_n = \bot$.}~\cite[Definition 5]{HeizmannHP2010} $J_0, \ldots, J_n$, where $J_i = \{ (\widebar{\vec{x}}, \vec{l}) \mid I_{q_i}(\widebar{\vec{x}}, \vec{l}, r) \}$ (we use $(\widebar{\vec{x}}, \vec{l})$ to denote the corresponding valuation) and $r$ is the state just before the last pending call in $\rho_1 \cdots \rho_i$, or $\bot$ if no such state exists. By the construction of the CHC, $\pre \Rightarrow q_0$ and $q_n \Rightarrow \post$. One can show that such an annotation implies $\{\pre\} \rho \{\post\}$.

  $\Leftarrow$: Prove by induction that given a derivation of $I_q(\widebar{\vec{x}}, \vec{l}, r)$, one can construct a nested word $\rho$ such that
  \begin{itemize}
    \item There is a partial run of $A$ over $\rho$ that ends at state $q$. Moreover, if $r = \bot$, then $\rho$ is well-matched; otherwise, $\rho$ has a pending call, $r$ is the state before the last pending call in this run.
    \item $\rho$ can be annotated with an inductive sequence of nested interpolants that starts with some $\{\nu\}$ such that $\nu \models \pre$ and ends with $\{ (\vec{x}, \vec{l}) \}$.
  \end{itemize}
  Then, from a derivation of $\bot$, one can extract a counterexample that witnesses $\{\pre\}P\{\post\}$ does not hold.
\end{proof}

\paragraph{Encode a VPG}

The semantics of a well-matched syntactic run can be given without an explicit stack, as a relation on valuations rather than stacks of valuations. Let $T$ be the set of all valuations. Unlike the direct encoding of an NWA, the ghost variables for function arguments are not needed.

\begin{itemize}
\item $\sem{S'}{\epsilon} := \{ (\nu, \nu) : \nu \in T \}$.
\item $\sem{S'}{\letter{assume $\phi$}} := \{ (\nu, \nu) : \nu \models \phi \}$.
\item $\sem{S'}{\letter{$x$ := $t$}} := \{ (\nu, \nu \oplus \{x \mapsto \nu(t)\}) : \nu \in T \}$.
\item $\sem{S'}{\letter{call $q$}\, \rho\, \letter{ret $q$}} := \{ (\nu, \nu \oplus \bigcup_y\{y \mapsto \mu'(y)\}) : (\mu, \mu') \in \sem{S'}{\rho} \land \bigwedge_x \mu(x) = \nu(x) \}$, where $\rho$ is well-matched and $x,y$ range over the input and output variables of $q$ respectively.
\item $\sem{S'}{\rho\tau} = \sem{S'}{\tau} \circ \sem{S'}{\rho}$, where $\rho, \tau$ are well-matched.
\end{itemize}

Let $X$ be all variables in the program. Given a well-matched VPG $G = (V, \tS, P, S)$, for each nonterminal $L \in V$, declare a predicate $P_L(X, X')$ that \emph{overapproximates} the semantics of all words produced by $L$. Each production in $P$ translates to a constrained horn clause:

\begin{itemize}
\item $L \to a L'$ translates to
  \[ \sem{S}{a}(X, X') \land P_{L'}(X', X'') \implies P_L(X, X''). \]
\item $L \to \letter{call $q$}\, L' \, \letter{ret $q$}\, L''$ translates to
  \begin{multline*}
    \left(\bigwedge_x x' = x\right) \land \left(\bigwedge_y y^{(3)} = y''\right) \land \left(\bigwedge_{z\neq y}z^{(3)} = z\right) \\
    \land P_{L'}(X', X'') \land P_{L''}(X^{(3)}, X^{(4)}) \implies P_L(X, X^{(4)})
  \end{multline*}
  where $x,y$ range over the input and output variables of $q$ respectively.
\item $L \to \epsilon$ translates to
  \[ \bigwedge_{z\in X} z' = z \implies P_L(X, X'). \]
\end{itemize}

To verify a safety property $p(X, X')$, add for all start symbol $L$ the clause
\[ P_L(X, X') \land \neg p(X, X') \implies \bot. \]

\subsection{Parameterized Lockstep} \label{app:cr-details}

Define $G_0 := \bigcup_{i=1}^n \left(G_i \cup \{ Y' \to \blue{c}W\blue{r}E : (Y \to \blue{c}W\blue{r}V) \in G_i \}\right) \cup \{ E \to \epsilon \}$. The rules are listed in Fig.~\ref{fig:ls-vpg}. The nonterminal symbols of the grammar $G$ of the $\vec{m}$-lockstep of $\Lang(G_1), \ldots, \Lang(G_n)$ have the form $[\vec{s},\vec{t},\prod_{i=0}^k W_i]$, where $\vec{s},\vec{t}$ satisfy the constraints in Section \ref{sec:cr}, $W_i$ is a nonterminal in $G_0$ for each $i$, and the length and maximum element of $\vec{s}$ are bounded by those of $\vec{m}$.
Moreover, to account for \textsc{Ls-Eps}, $[\vec{s},\vec{t},\prod_{i=1}^{j-1}W_i \times E \times \prod_{i=j}^k W_i]$ and $[\vec{s},\vec{t},\prod_{i=1}^k W_i]$ are treated as the same symbol for any $k > 1$.

\begin{figure}[h]
\begin{mathpar}
   \inferrule[G-LsOne]
   {(X \to \alpha) \in G_0}
   {([\vec{s},\vec{t},X] \to \alpha) \in G}
   \and
   \inferrule[G-LsInt]
   {\forall i\in [k].\, (W_i \to \alpha_i) \in G_0\\
   (U_j \to \blue{a} W_j) \in G_0\\
   \text{none of $\alpha_i$ where $i < j$ starts with an internal}}
   {\textstyle \left([\vec{s},\vec{t},\prod_{i=1}^{j-1} W_i \times U_j \times \prod_{i=j+1}^k W_i] \to \blue{a}[\vec{s},\vec{t},\prod_{i=1}^k W_i] \right) \in G}
   \and
   \inferrule[G-LsFstCall]
   {\forall i\in [k].\, (Y_i \to \blue{c_i} W_i \blue{r_i} V_i) \in G_0\\
   \vec{t} = \vec{0}}
   {\textstyle \left([\vec{s},\vec{t},\prod_{i=1}^k Y_i] \to \blue{c_1}\, [\vec{s},\dec_{\vec{s}}(\vec{t}),W_1 \times \prod_{i=2}^k Y'_i]\, \blue{r_1}\, [\vec{s},\vec{t},\prod_{i=1}^n V_i]\right) \in G}
   \and
   \inferrule[G-LsRstCall]
   {\forall i\in [k].\, (Y_i \to \blue{c_i} W_i \blue{r_i} V_i) \in G_0\\
   \vec{t} \neq \vec{0}\\
   j = \min\{i : t_i > 0\}}
   {\textstyle \left([\vec{s},\vec{t},\prod_{i=1}^k Y_i] \to \blue{c_j}\, [\vec{s},\dec_{\vec{s}}(\vec{t}),\prod_{i=1}^{j-1} Y'_i \times W_j \times \prod_{i=j+1}^k Y'_i]\, \blue{r_j}\, V_j \right) \in G}
\end{mathpar}
\Description{VPG construction rules of the parameterized lockstep}
\caption{VPG productions for parameterized lockstep \label{fig:ls-vpg}.  $[k]$ denotes the set $\{1,\ldots,k\}$.}
\end{figure}

\subsection{Implementation of VPG Canonical Reductions via a Queue}

\begin{algorithm}
\caption{Direct VPG reduction}\label{alg:direct-vpg}
\begin{algorithmic}
    \Require $G_i$'s are uniform, with disjoint nonterminals
    \Procedure{Reduction}{$G_1,\ldots,G_n$}
    \State $V' \gets \Call{\blue{Start-Symbols}}{\attrib{G_1}{start}, \ldots, \attrib{G_n}{start}}$
    \State $P \gets \bigcup_i \attrib{G_i}{productions}$
    \State $Q \gets \Call{Queue}{V'}$, $V \gets \Call{Set}{V'}$, $R \gets \Call{Vpg}{V'}$
    \While{$Q$ is nonempty}
        \State $X \gets \Call{Dequeue}{Q}$
        \State $(Y_1, \ldots, Y_m) \gets \Call{\blue{Components}}{X}$
        \ForAll{$(Y_1 \to \alpha_1), \ldots, (Y_m \to \alpha_m)$ in $P$}
            \State $\id{rhs}, \id{splits} \gets \Call{\blue{Product}}{X,\, ((Y_1 \to \alpha_1),\, \ldots,\, (Y_m \to \alpha_m))}$
            \State Add the production $X \to \id{rhs}$ to $R$
            \State Add the productions $\id{splits}$ to $P$
            \ForAll{nonterminal $B$ in $\beta$}
                \If{$B \notin V$}
                    \State $\Call{Enqueue}{Q,\, B}$
                    \State $\Call{Add}{V,\, B}$
                \EndIf
            \EndFor
        \EndFor
    \EndWhile
    \State \Return $R$
    \EndProcedure
\end{algorithmic}
\end{algorithm}

Algorithm \ref{alg:direct-vpg} outlines the scheme for constructing a (canonical) reduction as a VPG that maintains a queue $Q$ of reachable symbols, and is parametric on the implementations of the procedures $\Call{Start-Symbols}{}$, $\Call{Components}{}$, and $\Call{Product}{}$ which differ from one reduction to another. By maintaining a queue $Q$ of reachable symbols, Algorithm \ref{alg:direct-vpg} only adds those productions that are truly needed, making the construction efficient. We assume the input VPGs are \emph{uniform}: for any productions $X \to \alpha$ and $X \to \beta$, the right-hand sides $\alpha, \beta$ both start with a call, both start with an internal, or both are empty. Any VPG can be converted in linear time to an equivalent one that is uniform. See also Algorithm \ref{alg:direct-nestedconcat} and \ref{alg:direct-lockstep}.

\begin{algorithm}
\caption{Direct VPG reduction for $\nestedconcat$ of two components} \label{alg:direct-nestedconcat}
\begin{algorithmic}[1]
    \Function{Start-Symbols}{$V^{\mathsf{st}}_1, V^{\mathsf{st}}_2$}
    \State \Return $\{ [X_1, X_2] : X_1 \in V^{\mathsf{st}}_1, X_2 \in V^{\mathsf{st}}_2\}$
    \EndFunction
    \Statex

    \Function{Components}{$X$}
    \If{$X$ is of the form $[Y_1, Y_2]$}
        \State \Return $(Y_1, Y_2)$
    \Else
        \State \Return $(X)$
    \EndIf
    \EndFunction
    \Statex

    \Function{Product}{$X,\, \vec{p}$}
    \State $(Y_1 \to \alpha_1), \ldots, (Y_n \to \alpha_n) \gets \vec{p}$
    \If{n = 1}
        \State \Return $\alpha_1, \varnothing$
    \ElsIf{$\alpha_1$ is of the form $c Z_1 r W_1$}
    \State \Return $c[Z_1,Y_2]rW_1,\, \varnothing$
    \ElsIf{$\alpha_1$ is of the form $a Z_1$}
    \State \Return $a[Z_1,Y_2],\, \varnothing$
    \Else
    \State \Return $\alpha_2,\, \varnothing$
    \EndIf
    \EndFunction
\end{algorithmic}
\end{algorithm}
    
\begin{algorithm}
\caption{Direct VPG reduction for $\vec{s}$-lockstep} \label{alg:direct-lockstep}
\begin{algorithmic}[1]
    \Function{Start-Symbols}{$V^{\mathsf{st}}_1, \ldots, V^{\mathsf{st}}_n$}
    \State \Return $\{ [\vec{s},\vec{0},\vec{X}] : X_1 \in V^{\mathsf{st}}_1, \ldots, X_n \in V^{\mathsf{st}}_n\}$
    \EndFunction
    \Statex

    \Function{Components}{$X$}
    \If{$X$ is of the form $[\vec{s},\vec{t},\vec{Y}]$}
        \State \Return $\vec{Y}$
    \Else
        \State \Return $(X)$
    \EndIf
    \EndFunction
    \Statex

    \Function{Product}{$X,\, \vec{p}$}
    \State $(Y_1 \to \alpha_1), \ldots, (Y_n \to \alpha_n) \gets \vec{p}$
    \If{$n = 1$}
        \State \Return $\alpha_1,\, \varnothing$
    \EndIf
    \State $[\vec{s},\vec{t},\vec{Y}] \gets X$
    \If{$n > 1$ and there is $i$ such that $\alpha_i = \epsilon$}
        \State drop the $i$-th component in $\vec{s},\vec{t},\vec{Y},\vec{p}$
        \State \Return $\Call{Product}{[\vec{s},\vec{t},\vec{Y}],\, \vec{p}}$
    \ElsIf{there is a smallest $i$ such that $\alpha_i$ is of the form $a Z_i$}
        \State $\vec{Y'} \gets \vec{Y}[i \mapsto Z_i]$
        \State \Return $a[\vec{s},\vec{t},\vec{Y'}],\, \varnothing$
    \Else
        \State write $\alpha_i = c_i Z_i r_i W_i$ for each $i$
        \State $\vec{t'} \gets (\vec{t} - 1) \bmod \vec{s}$
        \If{$\vec{t} = \vec{0}$}
            \State let $U_i = [c_iZ_ir_i]$ for each $i$
            \State $\vec{Z'} \gets \vec{U}[1 \mapsto Z_1]$
            \State \Return $c_1[\vec{s},\vec{t'},\vec{Z'}]r_1[\vec{s},\vec{t},\vec{W}],\, \{ U_i \to c_iZ_ir_iE : i = 2, \ldots, n \} \cup \{ E \to \epsilon \}$
        \Else
            \State $i \gets \min\{j : t_j > 0\}$
            \State $\vec{Z'} \gets \vec{Y}[i \mapsto Z_i]$
            \State \Return $c_i[\vec{s},\vec{t'},\vec{Z'}]r_iW_i,\, \varnothing$
        \EndIf
    \EndIf
    \EndFunction
\end{algorithmic}
\end{algorithm}

\section{Benchmarks} \label{app:benchmarks}

Table \ref{tab:relrecmc-benchmarks}, \ref{tab:weaver-benchmarks}, \ref{tab:arith-benchmarks}, and \ref{tab:array-benchmarks} describe the benchmarks that we use for evaluating {\tool}. Table~\ref{tab:relrecmc-benchmarks} lists the benchmarks adapted from the safe relational benchmarks from \cite{MordvinovF2019}; if a benchmark is about recursive functions that are better understood by looking at the code, we list the name of the original benchmark and the number $k$ when it is viewed as a $k$-property; these recursive functions only call themselves and at most once. The Ackermann benchmark from \cite{MordvinovF2019} is discussed in Section~\ref{sec:customizations}. Table~\ref{tab:weaver-benchmarks} lists the benchmarks adapted from the iterative sequential benchmarks from~\cite{FarzanV2019}.

The benchmarks involve a variety of common functions on natural numbers and aggregate functions on arrays. We make the following assumptions:
\begin{itemize}
    \item The implementation of $\operatorname{mult}$ recurses on the first parameter.
    \item On arrays, the functions $\min, \max, \operatorname{sum}$ scan linearly from left to right unless otherwise specified. On binary trees simulated by arrays via a binary heap, they follow the structure of the heap.
    \item For arrays and trees, pointwise operations (sum and scaling) are fused with aggregate functions. For example, $\operatorname{sum}(A + 2B)$ is implemented as a single function that scans $A, B$ simultaneously from left to right.
\end{itemize}

For succinctness, we may omit ranges of variables that can be inferred from context, such as $n \ge 0$.

We label benchmarks 1--30, 32, 41--57, 77--99 and 103 with $(1, \ldots, 1)$-lockstep; 33--40, 58--73, and 99 with $\vec{s}$-lockstep where $\vec{s} \neq (1, \ldots, 1)$; 31, 74--76, and 100--102 with (1, 1)-lockstep of one component and the nested concatenation of the other two.

Note that for benchmarks that are theoretically out of reach for our tool, we still label them with $(1, \ldots, 1)$-lockstep for a uniform treatment.

\newcounter{rownum} 
\newcommand{\row}{\stepcounter{rownum}\therownum}

\begin{table}[h]
    \centering
    \begin{tabular}{|c|l|}
        \hline
        \textbf{No.} & \textbf{Description} \\
        \hline
        \setcounter{rownum}{0}
        \row & \texttt{add-horn}, a 2-property \\
        \row & \texttt{barthe}, a 2-property \\
        \row & $\operatorname{div}(n + d, d) = \operatorname{div}(n, d) + 1$ \\
        \row & $d < d' \implies \operatorname{div}(n, d) \ge \operatorname{div}(n, d')$ \\
        \row & \texttt{double-inc-loop}, a 2-property \\
        \row & $n^n < n!$ for $n > 1$ \\
        \row & $n < n' \implies \operatorname{fib}(n) \le \operatorname{fib}(n')$ \\
        \row & $\operatorname{fib}(n) = \operatorname{fib}'(n + 1)$, where the first term in $\operatorname{fib}'$ is $\operatorname{fib}'(1)$ \\
        \row & $\operatorname{fib}$ is equivalent to $\operatorname{fib}'$, where the two recursive calls are swapped \\
        \row & $\operatorname{gcd}(n, m) = \operatorname{gcd}(m, n)$ \\
        \row & \texttt{inc-loop-1}, a 2-property \\
        \row & \texttt{inc-loop-2}, a 3-property \\
        \row & \texttt{inc-loop-5}, a 6-property \\
        \row & \texttt{limit-1}, a 2-property \\
        \row & \texttt{limit-2}, a 2-property \\
        \row & $1 < n < n' \implies \operatorname{lucas}(n) \le \operatorname{lucas}(n')$ \\
        \row & $n > 2 \implies \operatorname{fib}(n) < \operatorname{lucas}(n)$ \\
        \row & $\operatorname{McCarthy91}$ is equivalent to $\operatorname{McCarthy91}'$, where the two branches are swapped \\
        \row & $n \le n' \implies \operatorname{McCarthy91}(n) \le \operatorname{McCarthy91}(n')$ \\
        \row & $q = \operatorname{div}(n, d) \implies \operatorname{mult}(q, d) + \operatorname{mod}(n, d)$ \\
        \row & $\operatorname{mod}(n + d, d) = \operatorname{mod}(n, d)$ \\
        \row & $a < a', b > 0 \implies \operatorname{mult}(a, b) < \operatorname{mult}(a', b)$ \\
        \row & $0 < a < a', 0 < b < b' \implies \operatorname{mult}(a, b) < \operatorname{mult}(a', b')$ \\
        \row & the number of digits of $n$, i.e., $\lfloor \log_{10} n \rfloor + 1$, is monotone \\
        \row & $0 < x < y,\, n > 0 \implies x^n < y^n$ \\
        \row & $x, y > 0,\, n > 0 \implies x^n + y^n \le (x + y)^n$ \\
        \row & \texttt{sum-6}, a 7-property \\
        \row & \texttt{twice-sum-1}, a 2-property \\
        \hline
    \end{tabular}
    \caption{Benchmarks adapted from \cite{MordvinovF2019} \label{tab:relrecmc-benchmarks}}
\end{table}

\begin{table}[h]
    \centering
    \begin{tabular}{|c|l|}
        \hline
        \textbf{No.} & \textbf{Description} \\
        \hline
        \row & array equality is symmetric \\
        \row & array equality is transitive \\
        \row & $\operatorname{mult}(a, c) + \operatorname{mult}(b, c) = \operatorname{mult}(a + b, c)$ \\
        \row & $\operatorname{mult}(a, b) + \operatorname{mult}(a, c) = \operatorname{mult}(a, b + c)$ \\
        \row & equivalence before and after unrolling a loop twice \\
        \row & equivalence before and after unrolling a loop three times \\
        \row & equivalence before and after unrolling a loop four times \\
        \row & equivalence before and after unrolling a loop five times \\
        \row & equivalence before and after unrolling a loop twice, with cleaning up \\
        \row & equivalence before and after unrolling a loop three times, with cleaning up \\
        \row & equivalence before and after unrolling a loop four times, with cleaning up \\
        \row & equivalence before and after unrolling a loop five times, with cleaning up \\
        \row & the lexicographic order comparator is symmetric \\
        \row & a variation of the lexicographic order comparator is symmetric \\
        \row & the lexicographic order comparator respects equality \\
        \row & a variation of the lexicographic order comparator respects equality \\
        \row & the lexicographic order is transitive \\
        \hline
    \end{tabular}
    \caption{Benchmarks adapted from the iterative sequential benchmarks from \cite{FarzanV2019} \label{tab:weaver-benchmarks}}
\end{table}

\begin{table}[h]
    \centering
    \begin{tabular}{|c|l|}
        \hline
        \textbf{No.} & \textbf{Description} \\
        \hline
        \row & the function $\operatorname{inc}$ that recursively sums up $n$ ones is injective \\
        \row & $a > 0 \land b < b' \implies \operatorname{mult}(a, b) < \operatorname{mult}(a, b')$ \\
        \row & $\operatorname{div}(n, d) = \operatorname{div}(2n, 2d)$ \\
        \row & $n \le n' \implies \operatorname{div}(n, d) \le \operatorname{div}(n', d)$ \\
        \row & $n \le n', d \ge d' \implies \operatorname{div}(n, d) \le \operatorname{div}(n', d')$ \\
        \row & $2\, \operatorname{mod}(n, d) = \operatorname{mod}(2n, 2d)$ \\
        \row & $(1 + 2) + \cdots + ((2n-1) + 2n) = (2 + 4 + \cdots + 2n) + (1 + 3 + \cdots + (2n-1))$ \\
        \row & $m < n \implies m! \le n!$ \\
        \row & $a > 1,\, 0 < m < n \implies a^m <a^n$ \\
        \row & $\operatorname{fib}(n) = 3\, \operatorname{fib}(n - 3) + 2\, \operatorname{fib}(n - 4)$ for $n > 3$ \\
        \row & the Tribonacci sequence is monotone \\
        \row & the Tetranacci sequence is monotone \\
        \row & $2\, \operatorname{inc}(n) = \operatorname{inc}(2n)$ \\
        \row & $\operatorname{inc}(2n) = 2\, \operatorname{inc}(n)$ \\
        \row & $2\, \operatorname{mult}(a, b) = \operatorname{mult}(2a, b)$ \\
        \row & $\operatorname{mult}(2a, b) = 2\, \operatorname{mult}(a, b)$ \\
        \row & $3\, \operatorname{mult}(a, b) = \operatorname{mult}(3a, b)$ \\
        \row & $\operatorname{mult}(3a, b) = 3\, \operatorname{mult}(a, b)$ \\
        \row & $3\, \operatorname{mult}(2a, b) = 2\, \operatorname{mult}(3a, b)$ \\
        \row & $2\, \operatorname{mult}(3a, b) = 3\, \operatorname{mult}(2a, b)$ \\
        \row & $d \mid n \implies d \mid 2n$ \\
        \row & $d \mid 2n \impliedby d \mid n$ \\
        \row & $d \mid n \implies d \mid 3n$ \\
        \row & $d \mid 3n \impliedby d \mid n$ \\
        \row & $2\, \operatorname{div}(n, d) \le \operatorname{div}(2n, d)$ \\
        \row & $\operatorname{div}(2n, d) \ge 2\, \operatorname{div}(n, d)$ \\
        \row & $3\, \operatorname{div}(n, d) \le \operatorname{div}(3n, d)$ \\
        \row & $\operatorname{div}(3n, d) \ge 3\, \operatorname{div}(n, d)$ \\
        \row & $\operatorname{div}(n, d) + \operatorname{div}(n', d) \le \operatorname{div}(n + n', d)$ \\
        \row & $d \mid n,\, d \mid n' \implies d \mid (n + n')$ \\
        \row & $\operatorname{mod}(n, d) = 0,\, \operatorname{mod}(n', d) = 0 \implies \operatorname{mod}(n + n', d) = 0$ \\
        \row & $\operatorname{mult}(a, b) = \operatorname{mult}(b, a)$ \\
        \hline
    \end{tabular}
    \caption{More arithmetic benchmarks \label{tab:arith-benchmarks}}
\end{table}

\begin{table}[h]
    \centering
    \begin{tabular}{|c|l|}
        \hline
        \textbf{No.} & \textbf{Description} \\
        \hline
        \row & $\operatorname{sum}(A) + \operatorname{sum}(B) = \operatorname{sum}(A + B)$ \\
        \row & the lexicographic order is antisymmetric \\
        \row & $\operatorname{min}(A + B) \ge \operatorname{min}(A) + \operatorname{min}(B)$ \\
        \row & $\operatorname{min}(A) \le \operatorname{max}(A)$ for nonempty array $A$ \\
        \row & $2\, \operatorname{sum}(A) = \operatorname{sum}(A)$ \\
        \row & $2\, \operatorname{sum}(A) + \operatorname{sum}(B) = \operatorname{sum}(2\, A + B)$ \\
        \row & pointwise $\le$ on arrays is transitive \\
        \row & pointwise $\le$ on arrays is reflexive \\
        \row & pointwise $\le$ on arrays is antisymmetric \\
        \row & $A \le B \implies 2A \le 2B$ \\
        \row & $A \le B \implies A + C \le A + C$ \\
        \row & $\max(A) < \min(B) \implies \max(B) < \min(C) \implies \min(A) < \min(C)$ \\
        \row & $\max(A) < \min(B) \implies \max(B) < \min(C) \implies \min(A) < \max(C)$ \\
        \row & $\min(T) \le \max(T)$ for tree $T$ \\
        \row & $\min(T) \le \max(T)$ for tree $T$, another implementation of $\min$, $\max$ \\
        \row & $\max(T) < \min(T') \implies \min(T) < \min(T')$ \\
        \row & $T \le T' \implies 2T \le 2T'$, where $T, T'$ has the same shape and $\le$ is pointwise \\
        \row & $\operatorname{sum}(T) + \operatorname{sum}(T') = \operatorname{sum}(T + T')$ \\
        \row & $\min(A) \le \max(A)$, where $\min, \max$ recurse on the two halves of $A$ \\
        \row & $\max(A) < \min(B) \implies \min(A) < \min(B)$, where $\min, \max$ recurse on the two halves \\
        \row & $\min(T) + \min(T') \le \min(T + T')$ \\
        \row & sum up two consecutive elements at a time = sum up one element at a time \\
        \row & $\min(A[i..k]) = \min(\min(A[i..j]), \min(A[j..k]))$ \\
        \row & $\max(A[i..k]) = \max(\max(A[i..j]), \max(A[j..k]))$ \\
        \row & $\operatorname{sum}(A[i..k]) = \operatorname{sum}(A[i..j]) + \operatorname{sum}(A[j..k])$ \\
        \row & $i < i' < j' < i \implies \operatorname{sum}(A[i'..j']) < \operatorname{sum}(A[i..j])$ if every element of $A$ is positive \\
        \hline
    \end{tabular}
    \caption{More array benchmarks \label{tab:array-benchmarks}}
\end{table}

}{}

\end{document}